\documentclass[format=acmsmall,authorversion=true]{acmart}
\usepackage{acm-ec-26}

\AtBeginDocument{%
  \fancypagestyle{firstpagestyle}{%
    \fancyhf{}%
  }%
}

\usepackage{booktabs} 
\usepackage[ruled]{algorithm2e} 

\SetAlFnt{\small}
\SetAlCapFnt{\small}
\SetAlCapNameFnt{\small}
\SetAlCapHSkip{0pt}
\IncMargin{-\parindent}

\usepackage{hyperref}
\usepackage{amsthm}
\AtEndPreamble{%
\newtheorem{claim}[theorem]{Claim}
\theoremstyle{acmdefinition}
\newtheorem{remark}[theorem]{Remark}
}

\setcitestyle{authoryear}

\usepackage[utf8]{inputenc}
\usepackage{tikz}
\tikzset{node/.style = {circle, fill, minimum size=4pt, inner sep=0pt}}
\usepackage{color}
\usepackage{enumitem}

\usepackage{amsmath}
\usepackage{bbm}

\usepackage[capitalize]{cleveref}
\crefname{claim}{Claim}{Claims}

\usepackage{thm-restate}

\newcommand{\N}{\mathbb{N}}
\newcommand{\R}{\mathbb{R}}
\newcommand{\C}{\mathcal{C}}

\newcommand{\Xflow}{X^{\textup{flow}}_{G,s,t}}
\newcommand{\Xpath}{X^{\textup{path}}_{G,s,t}}

\DeclareMathOperator{\dep}{dep}

\DeclareMathOperator{\supp}{supp}
\DeclareMathOperator{\aff}{aff}
\DeclareMathOperator{\lin}{lin}
\DeclareMathOperator{\conv}{conv}

\title[Controllability of Combinatorial Optimization Problems]{When to Identify Is to Control: On the Controllability of Combinatorial Optimization Problems}

\author{Max Klimm}
\affiliation{%
\institution{TU Berlin}
  \country{Germany}}
\author{Jannik Matuschke}
\affiliation{%
\institution{KU Leuven}
  \country{Belgium}}

\author{}

\begin{abstract}
Consider a finite ground set $E$, a set of feasible solutions $X \subseteq \mathbb{R}^{E}$, and a class of objective functions $\mathcal{C}$ defined on $X$.
We are interested in subsets~$S$ of~$E$ that control~$X$ in the sense that we can induce any given solution~$x \in X$ as an optimum for any given objective function $c \in \mathcal{C}$ by adding linear terms to $c$ on the coordinates corresponding to $S$.
This problem has many applications, e.g., when $X$ corresponds to the set of all traffic flows, the ability to control implies that one is able to induce all target flows by imposing tolls on the edges in~$S$. 

Our first result shows the equivalence between controllability and identifiability. If $X$ is convex, or if $X$ consists of binary vectors, then $S$ controls $X$ if and only if the restriction of $x$ to $S$ uniquely determines $x$ among all solutions in $X$.
In the convex case, we further prove that the family of controlling sets forms a matroid. This structural insight yields an efficient algorithm for computing minimum-weight controlling sets from a description of the affine hull of $X$.

While the equivalence extends to matroid base families, the picture changes sharply for other discrete domains.
We show that when~$X$ is equal to the set of $s$-$t$-paths in a directed graph, deciding whether an identifying set of a given cardinality exists is $\Sigma\mathsf{_2^P}$-complete. The problem remains $\mathsf{NP}$-hard even on acyclic graphs. For acyclic instances, however, we obtain an approximation guarantee by proving a tight bound on the gap between the smallest identifying sets for $X$ and its convex hull, where the latter corresponds to the $s$-$t$-flow polyhedron.
\end{abstract}

\begin{document}

\begin{titlepage}

\maketitle

\vspace{0.5cm}
\setcounter{tocdepth}{2} 
\tableofcontents

\end{titlepage}

\section{Introduction}

A fundamental goal of the study of complex systems is the ability to steer them into a favorable state.
On an intuitive level, such a steering process can be phrased as the  ability to determine a relatively large number of parameters of the system by controlling a comparatively small number of them.
In this paper, we study this problem for the case where the system is described as the optimal solutions to a (combinatorial) minimization problem.
To describe the setting formally, let~$E$ be a finite ground set and let~$X \subseteq \R^E$ be a set of feasible states. Further, let~$\C$ be a set of cost functions~$c \colon X \to \R$ that map each state $x \in X$ to a cost value~$c(x) \in \R$.
Implicit in our work is the assumption that the system settles in one of the states
\begin{align}
\label{eq:original-minimum}
\arg\min \bigl\{c(x) : x \in X\bigr\}
\end{align}
minimizing the cost.

This is a suitable assumption, e.g., when $X$ describes the possible states of a physical system and~$c(x)$ is the energy of state~$x \in X$.
In this case, the assumption stipulates that the system settles in an energy-minimal state.
For a concrete application, consider the electric current in an electric network modeled as a directed graph $G = (V,E)$ with a given source node~$s$ and a given target node~$t$.
Let~$\smash{b \in \mathbb{R}^V}$ denote the corresponding balance vector of a unit flow, i.e., $b_v = -1$ if $v=s$, $b_v = 1$ if $v=t$, and $b_v = 0$ otherwise. 
Each edge~$e \in E$ has a \emph{resistance}~$R_e$.
Interpreting positive flow values as flow in the direction of the arc, and negative flow as flow against the direction of the arc, we see that the set $X$ of all feasible flows satisfying the given balances is given by the set
\begin{align*}
X = \bigl\{x \in \mathbb{R}^E : \textstyle\sum_{e \in \delta^-(v)} x_e - \sum_{e \in \delta^+(v)} x_e =  b_v \text{ for all } v \in V\bigr\},
\end{align*}
where $\delta^+(v)$ and $\delta^-(v)$ denote the set of edges leaving and entering node~$v$, respectively. 
By Thompson's principle \citep[see, e.g.,][]{doyle1984random}, the electric flow minimizes the total energy dissipation defined as $c(x) = \frac{1}{2}\sum_{e \in E} x_e^2 R_e$.

Another application domain are systems with selfish agents where equilibria can be expressed as the minimal of potential functions.
For example, consider the traffic flows in the non-atomic model of~\citet{wardrop1952some}.  Here, we are given a directed graph $G = (V,E)$ modeling a road network with flow-dependent travel time functions $\ell_e \colon \R_+ \to \R_+$.
In addition, there is a finite set of commodities~$K$ such that for every commodity~$i \in K$ a total traffic volume of $d_i \in \R_+$ travels from node $s_i \in V$ to node $t_i \in V$ in the network.
Let $\smash{b_i \in \R^V}$ denote the balance vector of commodity $i \in K$ on node~$v \in V$ defined as~$b_{v,i} = -d_i$ if $v = s_i$, $b_{v,i} = d_i$ if $v = t_i$, and $b_{v,i} = 0$, otherwise.
Then, the set of feasible states corresponds to all ways to route the traffic, i.e., 
\begin{multline*}
X = \bigl\{ x \in \R_+^E : x_e = \textstyle\sum_{i \in K} x_{e,i} \text{ with } \textstyle\sum_{e \in \delta^-(v)} x_{e,i} = \sum_{e \in \delta^+(v)} x_{e,i} = b_{v,i} \text{ for all } v \in V \text{ and } i \in K\bigr\}.
\end{multline*}
It is well-known \citep{wardrop1952some,beckman1956studies} that non-atomic traffic equilibria are the optimal solutions to the optimization problem $\min \bigl\{ \textstyle\int_{0}^{x_e} \ell_e(z)\,\mathrm{d}z : x \in X\bigr\}$.
Similarly, in an atomic congestion game, we are given a directed graph $G = (V,E)$ with travel time functions as before. In addition, there is a finite set of players~$N$ and each player~$i$ travels from node~$s_i \in V$ to node~$t_i \in V$ in the network. A strategy of player~$i$ is to choose a path from $s_i$ to $t_i$ in the network. The set of induced load vectors is given by
\begin{align*}
X = \bigl\{ x \!\in\! \mathbb{R}^E_+ : x_e \!=\! \textstyle\sum_{i \in N} x_{e,i} \text{ with $x_{i} \!=\! (x_{e,i})_{e \in E}$ is the incidence vector of an $s_i$-$t_i$-path for all $i \!\in\! N$}\bigr\}.
\end{align*}
It is well-known \citep{rosenthal1973class,monderer1996potential} that the minima of the potential function defined as $\smash{c(x) = \sum_{e \in E} \sum_{k=1}^{x_e} \ell_e(k)}$ are pure Nash equilibria of the game.

In this paper, we consider the natural task of a system designer to assert control over an as small as possible subset $S \subseteq  E$ of the ground set in order to steer the system into \emph{all possible states}.
More formally, assume that a feasible set~$X$ and a set of  cost functions~$\C \subseteq \{c \mid c \colon X \to \mathbb{R}\}$ is given.
The task is to determine a subset $S \subseteq E$ of minimum cardinality such that for \emph{every} target state $x^* \in X$ and for \emph{every} cost function $c \in \C$ there is a vector $\gamma \in \R^S$ such that
\begin{align}
\label{eq:parametric-minimum}
\tag{$P_{\gamma}$}
x^* \in \arg\min \bigl\{c(x) + \textstyle\sum_{e \in S} \gamma_e x_e : x \in X\bigr\}.
\end{align}
We are further interested in a weighted variant of the problem where we are given a weight~$w_e \in \mathbb{R}_+$ for each element $e \in E$ and the task is to find $S \subseteq E$ such that for every state $\smash{x^* \in X}$ and for every cost function $c \in \C$ there is a vector $\smash{\gamma \in \R^S}$ satisfying \eqref{eq:parametric-minimum}, and the total weight $\sum_{e \in S} w_e$ is minimized among all such sets.

We note that the intervention of adding a term of $\gamma_e x_e$ for the elements $e \in S$ to the cost value of a solution~$x$ has a sensible interpretation in many applications.
For electric networks, e.g., a flow minimizing $\smash{\frac{1}{2}\sum_{e \in E} x_e^2 R_e + \sum_{e \in S} \gamma_e x_e}$ corresponds to an electric flow where the characteristic curve of the resistors at edge~$e$ has been shifted by $\gamma_e$ which can be achieved by adding appropriate diodes.
Thus, the set~$S$ corresponds to a minimal placement of variable diodes in an electric network that is able to induce every flow in the network as an electric flow.
For non-atomic traffic equilibria, e.g., the set~$S \subseteq E$ corresponds to a subset of the edges where a tollbooth (or, in more modern road toll implementations, a license plate reader) needs to be installed.
Setting $\gamma_e = 0$ for all $e \notin S$, a flow minimizing
\begin{align*}
\sum_{e \in E} \Biggl(\int_{0}^{x_e} \ell_e(z)\,\mathrm{d}z + x_e \,\gamma_e \Biggr) = \sum_{e \in E} \int_{0}^{x_e} \bigl( \ell_e(z) + \gamma_e \bigr) \,\mathrm{d}z
\end{align*}
effectively corresponds to an equilibrium where the travel time functions of the edges in~$S$ are shifted by $\gamma_e$.
This has the interpretation of a traffic equilibrium where each flow particle has the same trade-off between time and money and has to pay~$\gamma_e$ to the system when traversing an edge~$e \in S$.\footnote{We here also allow for negative tolls which effectively correspond to a subsidy when using a road segment, this has been done before, e.g., by \citet{Colini-Baldeschi18}.}
Similarly, for atomic congestion games, a flow minimizing
\begin{align*}
\sum_{e \in E} \sum_{k=1}^{x_e} \ell_e(k) + x_e \gamma_e = \sum_{e \in E} \sum_{k=1}^{x_e} \bigl( \ell_e(k) + \gamma_e\bigr)
\end{align*}
corresponds to an equilibrium where the travel time functions of the edges in~$S$ are shifted by~$\gamma_e$, so that the same interpretation as for non-atomic traffic networks applies.

In all the applications above it is natural to ask for a small set~$S$ (in the unweighted variant) or a weight-minimal set $S$ (in the weighted variant) as the installation and maintenance of tollbooths or variable diodes is costly, and, hence, a natural objective is to minimize these costs while being able to control the state.

We stress that we require that $S$ is able to induce \emph{every} state $x^*$ for \emph{every} cost function~$c \in \C$.
This is in contrast to a large part of the traffic toll literature where it is only required that for a single given cost structure $c$ a single given state (often a state minimizing a certain secondary objective) is reached.
Our approach allows for a more robust control where favorable states are still guaranteed to be enforceable even when the cost structure changes, e.g., because of technological changes or topological changes in the network.

\subsection{Our results}

For a set $X$ of feasible states and a set $\mathcal{C}$ of functions $c \colon X \to \mathbb{R}$, let us call a subset $S \subseteq E$ of elements \emph{controlling} for $(X,\mathcal{C})$ if for every $x^* \in X$ there exists a vector~$\smash{\gamma \in \mathbb{R}^S}$ such that $x^* \in \arg\min \bigl\{ c(x) + \sum_{e \in S} \gamma_e x_e : x \in X \bigr\}$.
The problem of finding controlling sets of small size is particularly challenging due to its bilevel nature: Even given a candidate set~$S$ one needs to verify that for every $x^* \in X$ and every $c \in \C$ there is a vector $\gamma$ such that $x^*$ appears as an optimal solution of \eqref{eq:parametric-minimum}.
To make the problem more tractable, we first give a characterization of controlling sets in terms of a different property. 
Let us call a subset $S \subseteq E$ of elements \emph{identifying} if each pair of distinct feasible states $x, x' \in X$ differs on $S$.

In \S~\ref{sec:control=identification},
we first show (\Cref{thm:convex}) that for convex sets~$X$ and convex cost functions~$\mathcal{C}$ the notions of controlling and identifying sets coincide.
The proof of this result further gives a characterization of the identifying (and hence, controlling) sets in terms of a certain matroid.
In particular, this allows us to compute minimum-weight controlling sets when a description of the affine hull of~$X$ is given.
We further show that the assumption that $\mathcal{C}$ contains only convex functions is tight: For a convex set~$X$, if $\mathcal{C}$ contains a single non-convex function, then there is no controlling set for $(X,\mathcal{C})$.

Next, we turn to the case where $X  \subseteq \{0,1\}^E$ consists of binary vectors.
Note that this includes all cases where the set of states is a family of subsets of $E$, such as, e.g., the paths in network or the spanning trees of a graph.
Also in this case, the notions of identifying and controlling sets coincide.
For the case that $X \subseteq \mathbb{Z}^E$ is a general set of integral vectors however, there may be identifying sets that are not controlling.

Motivated by the applications mentioned above involving flows and paths in networks, we turn to their controllability in \S~\ref{sec:controlling-paths-flows}.
We first give a characterization of the identifying sets for $s$-$t$-flows in \Cref{thm:flows} in terms of the cographic matroid.
This characterization further allows to compute minimum weight identifying sets for flows.
For the case that $X$ is the set of $s$-$t$-paths in a network without directed cycles, we provide a tight bound on the gap between the size of a minimum identifying set for such paths and for $s$-$t$-flows (which comprise the convex hull of the paths).
This yields a $\sqrt{m}$-approximation for finding minimum-size identifying sets for $s$-$t$-paths in acyclic networks, where $m$ is the number of arcs.
We also show that this problem is hard to approximate by a factor better than $2$.
For the general case in which the network may contain cycles, we derive significantly stronger hardness results, showing that the problem is $\Sigma\mathsf{_2^P}$-complete, that it is $\mathsf{NP}$-hard to verify whether a given set is identifying, and that no approximation factor significantly better than the trivial factor of $|E|$ can be obtained, unless~$\mathsf{P}=\mathsf{NP}$.

In \S~\ref{sec:matroids}, we then consider the case where $X$ is the set of bases of a matroid or the set of points of a polymatroid base polyhedron.
For both cases, we show that identifying (and thus controlling) sets correspond to a partition matroid defined by the connected components of the matroid/polymatroid.
Thus, again we can compute minimum-weight identifying sets in this case.

In \S~\ref{sec:explicit}, we finally consider the case where $\smash{X \subseteq \{0,1\}^E}$ is given as an explicit list of binary vectors (i.e., an explicitly given set of subsets of $E$).
For this case, we derive a $2\ln|X|$-approximation algorithm for the problem of computing an identifying (and hence, controlling) set of minimal weight, which we also show to be best possible, up to a constant factor, unless $\mathsf{P}=\mathsf{NP}$.

\subsection{Significance}

Our results provide a complete picture concerning the controllability of states for various systems of practical interest. In this section, we discuss the consequences for the controllability of the systems mentioned above.

For the problem of controlling an electrical network, \Cref{thm:convex} shows that the notion of controlling sets and identifying sets coincide, \Cref{thm:flows} gives a complete characterization of the identifying sets, and \Cref{cor:flows} yields a polynomial-time algorithm for computing an identifying set minimizing the weight of its elements.
In addition, the notion of identifying sets is meaningful for electric networks as well: These are sets of edges where measuring the electrical current and the corresponding connections allows to reconstruct the whole current in the network.

These positive results directly translate to non-atomic traffic networks with a single source~$s$ and a single target~$t$.
In addition, our characterization of identifying sets in \Cref{thm:flows} further allows to identify locations for tollbooths also in a setting with multiple sources and sinks under mild conditions on the network. Indeed, for networks where each edge~$e$ is contained in a directed cycle (which is clearly satisfied by real-world road networks), choosing the subset of directed edges~$S$ such that after the removal of~$S$ there is no undirected cycle in the network, is necessary to be controlling \emph{even for a single source--target pair}. Conversely, it is also sufficient to control flow between all source--target pairs.
More specifically, by charging a different toll $\gamma_{i,e}$ for each commodity~$i \in K$ and each edge~$e \in S$, all multi-commodity flows can be induced.

Previous work also considered non-atomic congestion games where for each commodity~$i \in K$, the total demand is split among the bases of a commodity-specific matroid rather than paths in a network \citep{FujishigeGHPZ17}. For single-commodity instances, the equilibria in this game can be obtained as the minimum of a convex function over a polymatroid base polytope.
For this case, \Cref{thm:polymatroid} establishes that the notions of controlling and identifying sets coincide and provides a complete characterization for these sets.
In addition, \Cref{cor:polymatroid} yields a polynomial-time algorithm for their computation.

Further, for toll enforcement in the systems above, it is also be important to figure our for a single non-atomic player which strategy it took in the network.
For single-commodity non-atomic congestion games, this problem is equal to finding an identifying set. \Cref{thm:path-general-digraphs-hard} implies that there is no $O(|E|^{1-\epsilon})$-approximation for finding a size-minimal identifying set for general directed graphs
and the $\Sigma\mathsf{^P_2}$-hardness result of \Cref{thm:paths-sigma-2-p} indicates a serious  barrier for practical computation, e.g., by ruling out compact IP formulations for the problem.
When the network does not contain directed cycles (as, e.g., the case when considering time-expanded networks in the context of dynamic flows), \Cref{thm:paths-hardness-DAGs} still implies that size-minimal identifying sets are hard to compute exactly, but \Cref{thm:paths-flow-gap} yields a tractable approximation for the problem.

\subsection{Related Work}

The controllability of traffic networks has been a longstanding research question.
The results of \citet{beckman1956studies} imply that the system-optimal flow can be obtained as the equilibrium when charging tolls on all edges.
\citet{cole2003pricing} showed that in single-commodity networks with heterogeneous players who value the trade-off between travel time and tolls differently, the system optimal flow can be obtained.
This result has been generalized to arbitrary networks \citep{fleischer2004tolls,karakostas2004edge,yang2004multi}.
These works obtain the toll prices from the dual information for an appropriate optimization problem; \citet{marcotte2009existence} gave a shorter proof of this result.
\citeauthor{fleischer2004tolls} also introduce the concept of \emph{enforceability}: they call a load vector \emph{enforceable} if there are tolls so that the load vector is obtained in equilibrium and give a characterization of enforceable load vectors.
\citet{harks2023unified} studied the enforceability of load vectors in a more general context that also includes market equilibria. 

The results above make no restriction on the set of elements that can have tolls.
\citet{verhoef2002second} studied a heuristic to minimize the total travel when tolls are only allowed on a given subset of edges.
\citet{hoefer2008taxing} provided a thorough complexity-theoretic analysis of this problem when the travel times are affine functions of the flow.
\citet{harks2015computing} generalized some of their positive results to non-affine travel time functions.
\citet{bonifaci2011efficiency} studied a more general model where tolls have to obey given edge-specific upper bounds.

\citet{hearn1998solving} introduced the minimum tollbooth problem, i.e., the problem of minimizing the cardinality of the edges that have a toll such that the system-optimal flow is enforced. 
\citet{bai2010heuristic} showed that this problem is $\mathsf{NP}$-hard and developed a metaheuristic for its solution.
\citet{bai2009combinatorial} proposed Benders cuts for a mixed-integer formulation of the problem.
\citet{basu2015new} strengthened this hardness result showing that the problem is $\mathsf{NP}$-hard to approximate within a factor of $1.13$ even for single-commodity instances with affine travel times. They further gave a polynomial algorithm for the special case that the network is single-commodity and series-parallel.
The minimum tollbooth problem has further been studied for atomic games. \citet{nickerl2021minimum} showed that finding a toll vector with minimal support such that there is a corresponding Nash equilibrium that is system-optimal is $\mathsf{W[2]}$-hard. 

While our work implicitly asks for a support of tolls that is robust under different states that need to be enforced and different cost structures, \citet{Colini-Baldeschi18} asked for tolls that are robust under different demand values.
A similar notion of tolls that are robust under different demand scenarios has been explored by \citet{ChristodoulouMP14} with the difference that here the tolls may be flow-dependent and count in the objective. Very recently, this problem has also been studied in a setting where the algorithm has access to a prediction for the demand \citep{christodoulou2025improving}.

Controllability is also closely related to inverse optimization, in which one seeks an objective function that renders a given solution optimal; see the recent survey by \citet{chan2025inverse}.
Specifically, in the context of combinatorial optimization problems, \citet{ahuja2001inverse} and \citet{berczi2025newton,berczi2023inverse} study the problem of finding such an objective function minimizing deviation with respect to a given norm from a nominal objective for shortest path, bipartite matching, and arborescence problems.
The controllability of systems has further been studied in the complex systems literature, e.g., by
\citet{liu2011controllability}.
Here, the dynamics of the system is more complicated and governed by differential equations rather than just the minimum of a cost function.

Identification problems, which we show to be largely equivalent to control problems in our context, arise also in the context of \emph{separating codes} in graphs~\citep{chakraborty2025open}, i.e., sets of nodes such that the neighborhood of a node inside the set suffices to identify the node.
When casting these problems into our framework of identifying sets, the ground set is the set of nodes and the set of solutions corresponds to the neighborhoods of the nodes in the graph. In particular, there is only a polynomial number of different solutions, and hence our logarithmic approximation results from \cref{sec:explicit} can be applied in this context.

\section{Preliminaries}
For $k \in \N$, we denote by $[k] := \{1,\dots,k\}$ the set of the first $k$ natural numbers.
Let $E$ be a finite set and let $X \subseteq \mathbb{R}^E$ be a non-empty set of feasible solutions.
Let $\mathcal{C}$ be a class of cost functions~$c \colon X \rightarrow \mathbb{R}$.
The following property is the central object of study of this paper.

\begin{definition}
A set $S \subseteq E$ \emph{controlling} for $(X, \mathcal{C})$ if for any $c \in \mathcal{C}$ and any $x^{\ast} \in X$ there is $\gamma \in \mathbb{R}^S$ such that the minimum $\min_{x \in X} \bigl\{c(x) + \sum_{e \in S} \gamma_e x_e\bigr\}$ exists and is attained at $x^{\ast}$.
\end{definition}

  Let $x_1,\dots,x_k \in X$ be a finite set of vectors and $\lambda_1, \dots, \lambda_k \in \R$. Then, $\smash{x = \sum_{i=1}^k \lambda_i x_i}$ is called a \emph{linear combination} of the vectors $x_1,\dots,x_k$. A linear combination is called an \emph{affine combination} if~$\smash{\sum_{i=1}^k \lambda_i = 1}$, and an affine combination is called a \emph{convex combination} if $\lambda_i \geq 0$ for all $i \in [k]$.
We denote by $\lin(X)$, $\aff(X)$, and $\conv(X)$ the linear hull, affine hull, and convex hull of $X$, i.e., the set of vectors that can be written as the linear, affine, or convex combination of finitely many vectors from $X$.
A set $X$ is a \emph{linear subspace} if $\lin(X) =X$, an \emph{affine subspace} if $\aff(X) =X$, and \emph{convex} if $\conv(X) = X$.
A set of vectors $\smash{\{x_0,\dots,x_k\} \subseteq X}$ is called \emph{affinely independent} if $x_1 - x_0,\dots,x_k - x_0$ are linearly independent.
In addition, $\{x_0,\dots,x_k\} \subseteq X$ is called an \emph{affine basis of $X$} if it is an inclusion-wise maximal affinely independent set.

A function $f : 2^E \to \R_+$ is called \emph{submodular} if $f(S \cap T) + f(S \cup T) \leq f(S) + f(T)$ for all~$S,T \in 2^E$. It is called \emph{monotone} if $f(S) \leq f(T)$ for all $S \subseteq T \subseteq E$, and \emph{normalized} if $f(\emptyset) = 0$.
For a set~$S \subseteq E$ and a vector $\smash{x \in \R^E}$, we write $x(S) = \sum_{e \in S} x_e$.
For a submodular, monotone, and normalized function~$f$, the set
\begin{align*}
B(f) = \bigl\{ x \in \mathbb{R}_+ : x(T) \leq f(T) \text{ for all $T \subseteq E$}, x(E) = f(E) \bigr\}
\end{align*}
the \emph{base polyhedron} associated with the polymatroid $(E,f)$. The function~$f$ is also called the \emph{polymatroid rank function}.

When a submodular, monotone, and normalized function in addition satisfies the properties that~$f(S) \in \N$ and $f(S) \leq |S|$ then it is the rank function of a \emph{matroid}.
The matroid is then defined as~$\mathcal{M}(f) = (E,\mathcal{I})$ where
\begin{align*}
\mathcal{I} = \bigl\{T \in 2^E : |T| \leq f(T) \bigr\}.
\end{align*}
For a matroid~$\mathcal{M} = (E,\mathcal{I})$, the sets~$S \in \mathcal{I}$ are called \emph{independent}, and the sets in $2^E \setminus \mathcal{I}$ are called \emph{dependent}.
A \emph{basis} of $\mathcal{M}$ is a maximal independent set, i.e., a set $S \in \mathcal{I}$ such that $S \cup \{e\} \notin \mathcal{I}$ for all~$e \in E \setminus S$.
A \emph{circuit} of $\mathcal{M}$ is a minimal dependent set, i.e., a set $S \in 2^E \setminus \mathcal{I}$ such that $S \setminus \{e\} \in \mathcal{I}$ for all $e \in S$.
For a matroid~$\mathcal{M}$, we also write 
$\operatorname{rk}(\mathcal{M})$ for the rank of that matroid, i.e., for the maximal size of an independent set of $\mathcal{M}$.

\section{Characterization of Controlling Sets}
\label{sec:control=identification}

In this section, we give a characterization of controlling sets for convex sets $X$ and convex cost functions~$\C$.
We call a set $S \subseteq E$ \emph{identifying} for $\smash{X \subseteq \mathbb{R}^E}$ if $\smash{S \cap \{e \in E : x_e \neq x'_e\} \neq \emptyset}$ for all $x, x' \in X$ with $x \neq x'$.
We first show that under a mild condition on the cost functions every controlling set is also identifying.
Formally, we call $(X,\mathcal{C})$ \emph{non-degenerated} if for every $x,x' \in X$ with $x \neq x'$ there is $c \in \mathcal{C}$ such that $c(x) \neq c(x')$.

\begin{lemma}\label{lem:ctrl=ident}
    Let $E$ be finite and $X \subseteq \mathbb{R}^E$ and $\mathcal{C}$ be a set of cost functions such that $(X, \mathcal{C})$ is non-degenerated. 
    If $S \subseteq E$ is controlling then it is also identifying.
\end{lemma}

\begin{proof}
Let $S$ be controlling. For a contradiction, let us suppose that $S$ was not identifying, i.e., there are $x, x' \in X$ with $x \neq x'$ such that $S \cap \{e \in E : x_e \neq x_e'\} = \emptyset$.
The non-degeneracy of $(X, \mathcal{C})$ implies that there is $c \in \mathcal{C}$ such that $c(x) \neq c(x')$ and it is without loss of generality to assume that~$c(x) > c(x')$. 
Since $S$ is controlling, there is $\gamma \in \mathbb{R}^S$ such that $x \in \arg\min_{y \in X} \bigl\{c(y) + \sum_{e \in S} \gamma_e y_e \bigr\}$.
We obtain
\begin{align*}
\textstyle  c(x) + \sum_{e \in S} \gamma_e x_e - \Bigl(  c(x') + \sum_{e \in S} \gamma_e x_e' \Bigr) = c(x) - c(x') > 0,
\end{align*}
where we used that $x_e = x_e'$ for all $e \in S$.
Thus, $x$ is not a minimizer of $c(x) + \sum_{e \in S} \gamma_e x_e$ for any~$\gamma \in \mathbb{R}^S$, contradicting the assumption that $S$ is controlling.
\end{proof}

In the remainder of this section, we turn to two important special cases, convex sets (\cref{sec:controlling-convex}) and discrete $0$-$1$-sets (\cref{sec:controlling-discrete}), for which we show that the converse of \cref{lem:ctrl=ident} is also true, i.e., identifying sets are also controlling sets.

\subsection{Controlling Sets in Convex Optimization}
\label{sec:controlling-convex}

We proceed to give a characterization of controlling sets when $X$ is convex and $\C$ consists of the set of convex cost functions. We will show that in this case the minimal controlling sets are the bases of a certain matroid.
To define this matroid, let $\{x_0, \dots, x_k\} \subseteq X$ be an affine basis of $X$ and define $\smash{A(F) := \bigl\{x_i - x_0 : i \in \{1, \dots, k\}\bigr\} \cup \bigl\{\chi_e : e \in F\bigr\}}$ for $F \subseteq E$.
Let
$$\mathcal{A}_X := \bigl\{F \subseteq E : A(F) \text{ is linearly independent} \bigr\}.$$
Note that $\mathcal{A}_X$ is a matroid on $E$ with $\operatorname{rk}(\mathcal{A}_X) = |E| - k$, where $k$ is the \emph{dimension} of the affine hull of $X$.
We obtain the following theorem.

\begin{theorem}\label{thm:convex}
    Let $\mathcal{C}$ be the set of convex cost functions. Let $X \subseteq \mathbb{R}^E$ be a convex set and let $S \subseteq E$. Then the following statements are equivalent.
    \begin{enumerate}[label=(\roman*.)]
        \item $S$ is a controlling set for $(X,\mathcal{C})$.\label{thm:convex:ctrl}
        \item $S$ is an identifying set for $X$.\label{thm:convex:ident}
        \item $E \setminus S \in \mathcal{A}_X$.\label{thm:convex:basis}
    \end{enumerate}
\end{theorem}

Note that \ref{thm:convex:basis} is equivalent to $S$ containing a basis of $\mathcal{A}_X^{\star}$, the dual matroid of $\mathcal{A}_X$. Hence, the minimal controlling sets of $X$ are exactly the bases of the matroid $\mathcal{A}_X^{\star}$. We note further that~$\operatorname{rk}(\mathcal{A}^{\star}_X) = |E| - \operatorname{rk}(\mathcal{A}_X) = k$, which, as discussed above, is the dimension of the affine hull of~$X$.

\begin{proof}[Proof of \cref{thm:convex}]
If the dimension~$k$ of the affine hull of $X$ is $0$, $X$ contains a single point only, every set if controlling and identifying, and $\mathcal{A}_X$ is the free matroid where all sets are independent, so the theorem holds in this case. In the following, we assume that $k \geq 1$. 

    Note that \ref{thm:convex:ctrl} implies \ref{thm:convex:ident} by \cref{lem:ctrl=ident}. We will show that \ref{thm:convex:ident} implies \ref{thm:convex:basis} and that \ref{thm:convex:basis} implies~\ref{thm:convex:ctrl}.
    
    To show that \ref{thm:convex:ident} implies \ref{thm:convex:basis}, assume that $S$ is an identifying set for $X$ and assume by contradiction that $E \setminus S \notin \mathcal{A}_X$. By definition of $\mathcal{A}_X$ this implies that there is $\lambda \in \mathbb{R}^k$ and $\mu \in \mathbb{R}^{E \setminus S}$ with~$(\lambda, \mu) \neq 0$ such that 
    \begin{align}
   \textstyle    \sum_{i=1}^{k} \lambda_i (x_i - x_0) + \sum_{e \in E \setminus S} \mu_e \chi_e = 0.\label{eq:lin-dep}
    \end{align}
    Note that $\Delta := \sum_{i=1}^{k} \lambda_i (x_i - x_0) \in \mathbb{R}^E \setminus \{0\}$ because the vectors $\chi_e$ for $e \in F$ are linearly independent.
    Note further that \eqref{eq:lin-dep} implies that $\Delta_e = 0$ for all $e \in S$.
    Let~$x'$ be a point in the interior of~$\conv(\{x_0, \dots, x_k\}) \subseteq X$ which exists as $k \geq 1$.
    Then, by the convexity of $X$, there is $\varepsilon > 0$ such that $\smash{x'' := x' + \varepsilon \Delta \in X}$.
    But $x''_e - x'_e = \varepsilon \Delta_e = 0$ for all $e \in S$ and thus $S$ is not identifying for $X$.
    We conclude that \ref{thm:convex:ident} implies~\ref{thm:convex:basis}.
    
    Now assume that $E \setminus S \in \mathcal{A}_X$.
    Let $c \in \mathcal{C}$ and $x \in X$.
    Recall that for a function $c : X \to \R$ a vector~$g \in \R^n$ is a \emph{subgradient} at $x$ if
    \begin{align*}
    c(z) &\geq c(x) + g^\top (z-x) && \text{ for all $z \in X$}.
    \end{align*}
    The set of subgradients of $c$ at $x$ is called the \emph{subdifferential} and is denoted by $\partial c(x)$. When $c$ is convex, then for all $x \in X$, the subdifferential $\partial c(x)$ is non-empty. We further have that $x^*$ is a minimizer of $c$ if and only if $0 \in \partial c(x^*)$.

    Let $x \in X$ and $d \in \partial c(x)$ be arbitrary.
    Let $\gamma \in \mathbb{R}^E$ be a solution to the system
    \begin{align*}
        (x_i - x_0)^\top \gamma & = - d^\top (x_i - x_0) && \text{ for all $i \in [k]$},\\
        \gamma(e) & = 0 && \text{ for all $e \in E \setminus S$}.
    \end{align*}
    Note that such a solutions exists because the rows of the coefficient matrix of the system correspond to the vectors in $A(F)$, which are linearly independent because $E \setminus S \in \mathcal{A}_X$.
    As $d$ is a subgradient of $c$ at $x$, we have
    \begin{align*}
    c(z) &\geq c(x) + d^\top(z-x) && \text{ for all $z \in X$.}
    \end{align*}
    We then obtain
    \begin{align}
    \label{eq:subgradient-gamma}
    c(z) + \gamma^\top z \geq c(x) + \gamma^\top z + d^\top(z-x)
    \end{align}
    Using that $x_0,\dots,x_k$ is as affine basis of $X$, we can write $z  = x_0 + \sum_{i \in [k]} \lambda_i (x_i - x_0)$ and $x = x_0 + \sum_{i \in [k]} \mu_i (x_i - x_0)$ for appropriate coefficients $\lambda_i, \mu_i \in \R$, and we obtain
    \begin{align*}
    \gamma^\top z + d^\top(z-x) &= \gamma^\top x_0 + \sum_{i \in [k]} \lambda_i \gamma^\top (x_i - x_0) + d^\top x_0 + \sum_{i \in [k]} \lambda_i d^\top (x_i - x_0) \\
    &\quad  -d^\top x_0 - \sum_{i \in [k]} \mu_i d^\top (x_i - x_0) \\
    &= \gamma^\top x_0 + \sum_{i \in [k]} \mu_i \gamma^\top (x_i - x_0) \\
    &= \gamma^\top x.
    \end{align*}
    Plugging this equation into \eqref{eq:subgradient-gamma}, this yields
    $
    c(z) + \gamma^\top z \geq c(x) + \gamma^\top x$,
    i.e., $0$ is a subgradient of the function $x \mapsto c(x) + \gamma^\top x$.
    Hence $x$ is an optimal solution to $\min_{x \in X} \bigl\{\sum_{e \in E} c_e(x_e) + \sum_{e \in S} \gamma_e x_e \bigr\}$.
\end{proof}

Note that by \cref{thm:convex}, the identifying/controlling sets for $X$ form a matroid.
We can optimize linear functions over this matroid when given access to an independence oracle.
Such an oracle can be obtained if we are given the affine hull of $X$, e.g., in the form of an affine basis $\{x_0, \dots, x_k\} \subseteq X$, and we obtain the following immediate corollary.

\begin{corollary}
    There is an algorithm that, given a finite set $E$, a weight function $w : E \rightarrow \mathbb{R}$ and an affine basis $\{x_0, \dots, x_k\} \subseteq X$ of a convex set $X$, computes in polynomial time in $|E|$ a subset $S \subseteq E$ such that $S$ is an identifying set for $X$ of minimum weight $\sum_{e \in S} w_e$. 
\end{corollary}

We next observe that when $X$ is convex but $\C$ contains a cost function that is not convex, no controlling set exists.

\begin{lemma}
Let $\smash{X \subseteq \mathbb{R}^E}$ be convex, and let $\C$ contain a non-convex function~$\tilde{c} : X \to \R$. Then, there is no controlling set for $(X, \C)$.
\end{lemma}

\begin{proof}
It is well-known that if $X$ is a convex set, and $c \colon X \to \R$ is a function such that $\partial c(x) \neq \emptyset$ for all $x \in X$, then $c$ is convex~\citep[Lemma 3.1.6]{nesterov2004}.
As $\tilde{c}$ is non-convex and $X$ is non-empty, there thus must be some $x \in X$ such that $\partial \tilde{c}(x) = \emptyset$.
Now assume by contradiction that $X = E$ is controlling for $X$ and $\C$. Then there is $\gamma \in \R^E$ such that $x \in \arg\min_{y \in X} \{\tilde{c}(x) + \gamma^\top x\}$. Defining~$c \colon X \to \R$ as $x \mapsto \tilde{c}(x) + \gamma^\top x$, this implies that $0 \in \partial c(x)$, i.e.,
\begin{align*}
&& \tilde{c}(z) + \gamma^\top z &\geq \tilde{c}(x) + \gamma^\top x && \text{ for all $z \in X$} \\
\Leftrightarrow && \tilde{c}(z) &\geq \tilde{c}(x) - \gamma^\top (z - x) && \text{ for all $z \in X$}, 
\end{align*}
which is a contradiction as the latter inequality implies that $-\gamma$ is a subgradient of $\tilde{c}$ at $x$.
\end{proof}

\subsection{Controlling Sets in Discrete Optimization}
\label{sec:controlling-discrete}

We now turn to the case where $X$ is a discrete set of binary vectors.
We observe that in this case any identifying set is also controlling.

\begin{theorem}
\label{thm:identifying-is-controlling-discrete}
    Let $E$ be a finite ground set, $X \subseteq \{0, 1\}^E$. Then if $S$ is identifying it is also controlling for the class of all cost functions. 
\end{theorem}

\begin{proof}
    Let us assume $S$ is identifying.
    Let $c \in \mathcal{C}$ and $x^* \in X$.
    Let $M := 2 \cdot \max_{x \in X} |c(x)|$.
    Let~$\gamma_e := -M$ for all $e \in S$ with $x^*_e = 1$ and $\gamma_e := M$ for all $e \in S$ with $x^*_e = 0$.
    Then, for any~$x \in X$, we have~$\sum_{e \in S} \gamma_e (x_e - x^*_e) \geq M$
    because there is at least one $e \in S$ with $x_e \neq x^*_e$, and, as a consequence, we have~$c(x) + \sum_{e \in E} \gamma_e x_e \geq c(x^*) + \sum_{e \in S} \gamma_e x^*_e$.
    Thus, $x^*$ is an optimizer of $\min_{x \in X} \bigl\{ c(x) + \sum_{e \in S} \gamma_e x_e \bigr\}$.
\end{proof}

\begin{corollary}
    Let $E$ be a finite ground set, $X \subseteq \{0, 1\}^E$ and let $\mathcal{C}$ be a set of cost functions such that $(X, \mathcal{C})$ is non-degenerate. Then $S \subseteq E$ is identifying if and only if it is controlling for $X$ and $\mathcal{C}$. 
\end{corollary}

\begin{remark}
    We remark that in the binary setting, i.e., $X \subseteq \{0, 1\}^E$, the notion of identifying sets can be equivalently expressed in terms of the corresponding set system $\mathcal{F}(X) := \{F \subseteq E : \chi_F \in X\}$ whose incidence vectors are the elements of $X$.
    In this case, $S \subseteq E$ is identifying for $X$ if and only if~$(F \Delta F') \cap S \neq \emptyset$ for all $F, F' \in \mathcal{F}(X)$ with $F \neq F'$.
\end{remark}

We close this section by showing with an example that in general identifying sets need not be controlling, even when all vectors of $X$ have integer entries.

\begin{example}[Identifying set that is not controlling]
Let $E = \{1,2,3\}$ and consider the feasible set  
\begin{align*}
X &= \bigl\{x \in \mathbb{Z}^E : x_1, x_2 \in \{0, 1\}, x_3 = x_1 + 2x_2 \bigr\} = \bigl\{(0, 0, 0), (0, 1, 1), (1, 0, 2), (1, 1, 3) \bigr\}.
\end{align*}
We observe that $S = \{3\}$ is an identifying set, but we claim that $S = \{3\}$ is not controlling for $(X,\mathcal{C})$ when $\mathcal{C}$ is the set of all linear functions.
To see this claim, suppose that $S = \{3\}$ was controlling for all linear functions.
This implies in particular that for $c(x) = x_1 - x_2 + x_3$, there is $\smash{\gamma}$ such that $\smash{x^{(1)} = (0,0,0) \in \arg\min \bigl\{c(x) + \gamma x_3 \bigr\}}$.
Writing $\smash{x^{(2)} = (0,1,1)}$ and $\smash{x^{(4)} = (1,1,3)}$, we obtain the inequalities
\begin{align*}
c(x^{(1)}) &\leq c(x^{(2)}) + \gamma &&\Leftrightarrow & 0 &\leq -1 + \gamma, \\
c(x^{(1)}) &\leq c(x^{(4)}) + 3 \gamma &&\Leftrightarrow & 0 &\leq 3 \gamma,
\end{align*}
which is a contradiction.
\end{example}

\section{Controlling Flows and Paths in Networks}
\label{sec:controlling-paths-flows}

In this section, we assume that we are given a directed graph (digraph) $G = (V, E)$ together with two nodes $s, t \in V$.
We are interested in the controlling/identifying sets for the $s$-$t$-flows and $s$-$t$-paths in $G$, respectively, i.e., for
\begin{align*}
    \Xflow & := \bigl\{x \in \mathbb{R}_+^E : x \text{ is an $s$-$t$-flow of value $1$ in $G$} \bigr\}, \quad \text{and}\\
    \Xpath & := \bigl\{\chi_P :  \text{$P$ is an $s$-$t$-path in $G$} \bigr\},
\end{align*}
where $\chi_P \in \{0,1\}^E$ is the indicator variable for the (edge set of) path $P$, i.e., the entry of $\chi_P$ corresponding to $e \in E$ is $1$ if $e$ lies on $P$, and is equal to $0$ otherwise.

Note that $\smash{\Xflow}$ is convex and that $\smash{\Xpath \subseteq \{0, 1\}^E}$.
Suppose that $\mathcal{C}$ is a set of convex functions.
Hence, in both cases, by our results from \cref{sec:control=identification}, controlling and identifying sets coincide for the respective choice of $X$.
We will thus focus on computation of identifying sets for $\smash{\Xflow}$ and $\smash{\Xpath}$, respectively.

We first discuss identifying sets for $\smash{\Xflow}$ in \cref{sec:controlling-flows}, which by \cref{thm:convex} form a matroid and thus allow for efficient computation.
We then discuss $\smash{\Xpath}$ for the case that $G$ is acyclic in \cref{sec:paths-DAG}. Although the problem of computing minimum identifying sets is $\mathsf{NP}$-hard even in this case, we are able to provide a tight bound on the gap between minimum identifying sets for $\smash{\Xflow}$ and $\smash{\Xpath}$, yielding an approximation for the latter.
Finally, in \cref{sec:paths-general-digraphs}, we consider computing identifying sets for $\smash{\Xpath}$ in general digraphs, for which we obtain strong hardness results.

\subsection{Controlling $s$-$t$-Flows}
\label{sec:controlling-flows}

We obtain the following theorem as a special case of \cref{thm:convex} applied to convex set $\Xflow$.

\begin{theorem}
\label{thm:flows}
    Let $G = (V, E)$ be a digraph, $s, t \in V$ such that there is a $s$-$t$-path in $G$, and $$E' := \{e \in E : e \text{ is on an $s$-$t$-path or directed cycle in $G$}\}.$$
    Then $S \subseteq E$ is identifying for $\Xflow$ if and only if $(V, E' \setminus S)$ contains no undirected cycle.
\end{theorem}

Before we turn to the proof of \cref{thm:flows}, we discuss some implications of the theorem.
A consequence of \cref{thm:flows} is that the minimal identifying sets for $\Xflow$ are the complements of spanning forests of $(V, E')$, or, in other words, they are the bases of the co-graphic matroid of the underlying undirected graph of $(V, E')$.
Furthermore, note that $E'$ can be constructed in polynomial time from $G$:
Determining the arcs of $E$ that are on a cycle can be done by determining the strongly connected components of $G$.
Moreover, any edge $e = (v, w) \in E$ that is not contained on a cycle is in $E'$ if and only if there is an $s$-$v$-path and a $w$-$t$-path in $G$ (note that the fact that $e$ does not lie on a cycle implies that these paths are node-disjoint, and hence $e$ lies on an $s$-$t$-path ).
Hence, we have an efficient independence oracle for the matroid whose bases are the minimal identifying sets of $\Xflow$. 
We obtain the following corollary.

\begin{corollary}\label{cor:flows}
    There is an algorithm that given a digraph $G = (V, E)$, two nodes $s, t \in V$, and weights $w \in \mathbb{R}^E$, finds an identifying set $S$ of $\Xflow$ minimizing $\sum_{e \in S} w_e$ in polynomial time.
\end{corollary}

\begin{proof}[Proof of \cref{thm:flows}]
    Assume that $E' \setminus S$ does not contain an undirected cycle.
    Let $x, x' \in \Xflow$ with $x \neq \emptyset$. Note that $x_e = x'_e = 0$ for all $e \in E \setminus E'$ by construction of $E'$.
    Note that $x - x'$ is a circulation, and hence its support is the union undirected cycles in $E'$, and because $x \neq x'$, there must be at least one such undirected cycle $C$ in the support.
    Because $C \not\subseteq E' \setminus S$, there must be~$e \in S$ such that $x_e - x'_e \neq \emptyset$, i.e., $x_e \neq x'_e$. 
    Thus $S$ is identifying for $\smash{\Xflow}$.

    For the converse, let $S$ be an identifying set for $\smash{\Xflow}$.
    By contradiction assume there is an undirected cycle $C \subseteq E' \setminus S$.
    Because $C \subseteq E'$, every $e \in C$ appears on some $s$-$t$-path or directed cycle, and therefore for every $e \in C$ there exists a flow $\smash{x^{(e)} \in \Xflow}$ with $\smash{x^{(e)}_e > 0}$.
    Let $\smash{x := \frac{1}{|C|} \sum_{e \in C} x^{(e)}}$.
    Note that $x$ is a convex combination of points in $\smash{\Xflow}$ and hence $\smash{x \in \Xflow}$ and that moreover~$x_e > 0$ for all $e \in C$ by construction.
    Now traverse $C$ in an arbitrary direction and let $C^+ \subseteq C$ denote the set of arcs that are traversed in forward direction and let $C^- \subseteq C$ denote the set of arcs that are traversed in backwards direction.
    Let $\varepsilon := \min_{e \in C^-} x_e$ and define $x'$ by 
    \begin{align*}
        x'_e := 
        \begin{cases}
            x_e + \varepsilon & \text {for } e \in C^+,\\
            x_e - \varepsilon & \text {for } e \in C^-,\\
            x_e & \text {for } e \in E \setminus C.
        \end{cases}
    \end{align*}
    Note that $\smash{x' \in \Xflow}$, because $x'_e \geq 0$ for all $e \in E$ by choice of $\varepsilon$ and because $x'$ fulfills flow conservation as it arose from $x$ by augmenting flow along a cycle in the residual network of $x$.
    Because $x_e = x'_e$ for all $e \in E \setminus C \supseteq S$, we obtain a contradiction to $S$ being identifying for $\smash{\Xflow}$.
\end{proof}

\subsection{Controlling $s$-$t$-paths in directed acyclic graphs}
\label{sec:paths-DAG}

We now turn our attention to $\smash{\Xpath}$.
We first consider the special case where the digraph $G$ is a directed acyclic graph (DAG).
As a first positive result, we observe that it is possible to verify in polynomial time whether a given set $S \subseteq E$ is identifying for $\smash{\Xpath}$ in this case (whereas the same problem is $\mathsf{NP}$-hard in the case where $G$ is an arbitrary digraph, as we shall see in \cref{sec:paths-general-digraphs}).

\begin{theorem}\label{thm:paths-DAG-verification}
    There is a polynomial algorithm that, given a DAG~$G = (V, E)$, two vertices~$s, t \in V$, and a set $S \subseteq E$, decides whether $S$ is identifying for $\smash{\Xpath}$
\end{theorem}

\begin{proof}
    Without loss of generality, we may assume that for every $v \in V$ there is an $s$-$v$-path and a $v$-$t$-path in~$G$.
    If this is not the case for some $v \in V$, we can remove~$v$ from the graph without changing the set of $s$-$t$-paths.
    For $v \in V$, let $E_v$ denote the set of edges in $E \setminus S$ whose tail is reachable from $v$ in the digraph $G' := (V, E \setminus S)$.
    Note that each $E_v$ can be computed easily in linear time via DFS starting from~$v$.
    We show that the set $S$ is identifying if and only if there are no nodes $v, w \in V$ such that there are two internally node-disjoint disjoint $v$-$w$-paths in $G'$, which is the case if and only if $E_v$ is an arborescence rooted at~$v$ for every $v \in V$.
    
    First, assume that there are no nodes~$v, w \in V$ such that there are two internally node-disjoint~$v$-$w$-paths in $G'$.
    Consider any two distinct~$s$-$t$-paths $P, Q$ in $G$ and note that every weakly connected component in the symmetric difference $P \Delta Q$ is the union of two internally node-disjoint $v$-$w$-paths for some $v, w \in V$. Hence $(P \Delta Q) \cap S \neq \emptyset$ for any two such paths, i.e., $S$ is identifying.
    Conversely, assume that there are $v, w \in V$ such that $G'$ contains two internally node-disjoint $v$-$w$-paths, say $P'$ and $Q'$. Then let $R_{sv}$ and $R_{wt}$, respectively, be an $s$-$v$-path and a $w$-$t$-path, respectively, which exist by our earlier assumption.
    Note that the concatenated paths $P := R_{sv} \circ P' \circ R_{wt}$ and~$Q := R_{sv} \circ Q' \circ R_{wt}$ are both $s$-$t$-paths (because $G$ is a DAG and hence the concatenation cannot contain a cycle). Thus,~$P \Delta Q \subseteq P' \cup Q' \subseteq E \setminus S$, i.e., $S$ is not identifying.
\end{proof}

Unfortunately, despite \cref{thm:paths-DAG-verification}, finding a minimum-size identifying set for $\Xpath$ is $\mathsf{NP}$-hard even for DAGs. In fact, we obtain the following hardness-of-approximation result.

\begin{theorem}\label{thm:paths-hardness-DAGs}
    Assuming the unique games conjecture and $\mathsf{P} \neq \mathsf{NP}$, the following problem does not have a $(2 - \varepsilon)$-approximation for any $\varepsilon > 0$: 
    Given a DAG $G = (V, E)$, two vertices $s, t \in V$, find a minimum-size identifying set for $\smash{\Xpath}$.
\end{theorem}

\begin{proof}
    We establish the hardness result via a reduction from \textsc{Vertex Cover}: Given an undirected graph $\bar{G} = (\bar{V}, \bar{E})$, find a minimum-size set of vertices $U \subseteq \bar{V}$ such that every edge in $\bar{E}$ is incident to at least one vertex in $U$.
    Under the unique games conjecture, \textsc{Vertex Cover} does not admit a~$(2 - \varepsilon)$-approximation for any~$\varepsilon > 0$ unless $\mathsf{P} = \mathsf{NP}$~\citep{khot2008vertex}.

    We construct a digraph $G = (V, E)$ as follows.
    We set $\smash{\ell := \big\lceil \frac{4|E|}{\varepsilon} \big\rceil}$.
    The node set of $G$ is given by
    \begin{align*}
        V & := \bigl\{s, t\bigr\} \cup \bigl\{u_e : e \in \bar{E}\bigr\} \cup \bigl\{v_i : v \in \bar{V}, i \in [\ell] \bigr\}.
    \end{align*}
    Further define $E_s := \{(s, u_e) : e \in \bar{E}\}$, $E' := \{(u_e, v_i) : e \in \bar{E}, v \in e, i \in [\ell]\}$, and~$E_i := \{(v_i, t) : v \in \bar{V}\}$ for $i \in [\ell]$ and construct the edge set of $G$ as
    \begin{align*}
        E & := E_s \cup E'  \cup E_1 \cup \dots \cup E_{\ell}.
    \end{align*}
    See \cref{fig:paths-dag-hardness} for a depiction of the construction.

    We will use the following claim which we proof in \cref{app:hardness-DAGs}.
    
    \begin{figure}[tb]
    \begin{center}
    \begin{tikzpicture}[xscale=1.8,yscale=1]
    	
        \path (-3, 0) node[node,label=below:{$v_1$}] (v1) {}
            ++(1, 0) node[node,label=below:{$v_2$}] (v2) {}
            ++(1, 0) node[node,label=below:{$v_3$}] (v3) {};
            
        \draw (v1) edge node[above] {$e$} (v2);
        \draw (v2) edge node[above] {$f$} (v3);
    	
        \path node[node,label=below:{$s$}] (s) {}
            +(1, 1) node[node,label=below:{$u_e$}] (ue) {}
            +(1, -1) node[node,label=below:{$u_f$}] (uf) {}
            ++(2, 2.4) node[node,label=above:{$v_{1,1}$}] (v11) {}
            +(0, -0.3) node (v1m) {$\vdots$}
            ++(0, -0.8) node[node,label=below:{$v_{1,\ell}$}] (v12) {}
            ++(0, -1.2) node[node,label=above:{$v_{2,1}$}] (v21) {}
            +(0, -0.3) node (v2m) {$\vdots$}
            ++(0, -0.8) node[node,label=below:{$v_{2,\ell}$}] (v22) {}
            ++(0, -1.2) node[node,label=above:{$v_{3,1}$}] (v31) {}
            +(0, -0.3) node (v3m) {$\vdots$}
            ++(0, -0.8) node[node,label=below:{$v_{3,\ell}$}] (v32) {}
            ++(1, 2.4) node[node,label=below:{$t$}] (t) {};
        
        \draw[-latex] (s) edge (ue);
        \draw[-latex] (s) edge (uf);
        
        \draw[-latex] (ue) edge (v11);
        \draw[-latex] (ue) edge (v12);
        \draw[-latex] (ue) edge (v21);
        \draw[-latex] (ue) edge (v22);
        
        \draw[-latex] (uf) edge (v21);
        \draw[-latex] (uf) edge (v22);
        \draw[-latex] (uf) edge (v31);
        \draw[-latex] (uf) edge (v32);
        
        \draw[-latex] (v11) edge (t);
        \draw[-latex] (v12) edge (t);
        \draw[-latex] (v21) edge (t);
        \draw[-latex] (v22) edge (t);
        \draw[-latex] (v31) edge (t);
        \draw[-latex] (v32) edge (t);
        
    \end{tikzpicture}
    \end{center}

    \vspace*{-0.3cm}
    
    \caption{%
    \label{fig:paths-dag-hardness}%
    Example of the construction from the proof of \cref{thm:paths-hardness-DAGs}. 
    The undirected graph on the left depicts an instance of \textsc{Vertex Cover}.
    The digraph resulting from the construction is depicted on the right.
    The arc set~$\{(s, u_e), (v_{2,1}, t), \dots, (v_{2,\ell},t)\}$ is a minimum-size identifying set for the $s$-$t$-paths in the digraph, corresponding to the vertex cover $\{v_2\}$ for the undirected graph.
    }
    \end{figure}
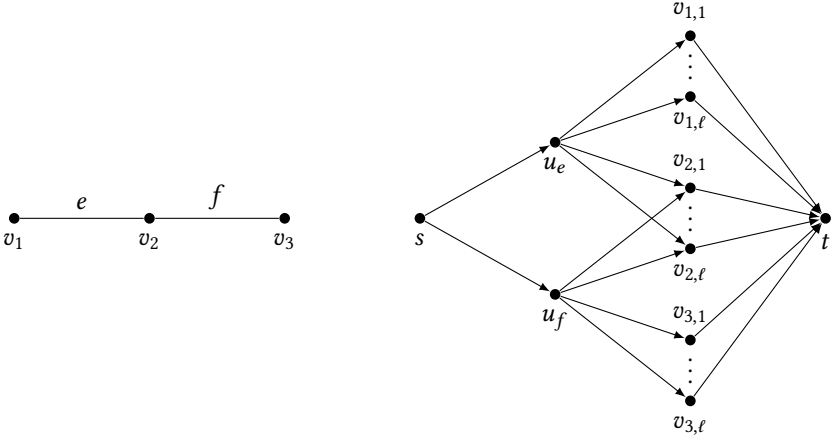

    \begin{restatable}{claim}{identifyingSetVertexCover}\label{clm:identifying-set-vertex-cover}
        There is a polynomial-time algorithm that, given an identifying set $S$ for $\smash{\Xpath}$, constructs a set $S'$ such that $|S'| \leq |S|$, 
        $S' \cap E' = \emptyset$, and $U_i := \{v \in \bar{V} : (v_i, t) \in S' \cap E_i\}$ is a vertex cover for each~$i \in [\ell]$.
    \end{restatable} 
    
    In particular, \cref{clm:identifying-set-vertex-cover} implies that the size of a minimum identifying set for $\smash{\Xpath}$ is at most $|E_s| + \ell k$, where $k$ is the size of a minimum vertex cover in $\bar{G}$. Thus, any $(2 - \varepsilon)$-approximation would yields an identifying set $S$ of size at most $(2 - \varepsilon)(|E_s| + \ell k)$.
    Applying \cref{clm:identifying-set-vertex-cover} to $S$ yields~$S'$ of size $|S'| \leq |S|$ such that $U_i := \{v \in \bar{V} : (v_i, t) \in S' \cap E_i\}$ is a vertex cover in $G$ for each~$i \in [\ell]$.
    In particular, $\sum_{i=1}^{\ell} |U_i| \leq (2 - \varepsilon)(|E| + \ell k)$.
    By the pigeonhole principle there must be an $i$ such that 
    $$|U_i| \leq \frac{(2 - \varepsilon)|E_s|}{\ell} + (2 - \varepsilon) k < \biggl(2 - \frac{\varepsilon}{2}\biggr) k,$$
    where the final inequality follows from $|E_s| = |E|$ and $\smash{\ell \geq \frac{4|E|}{\varepsilon}}$.
    Hence a $(2 - \varepsilon)$-approximation for finding a minimum-size identifying set for $\smash{\Xpath}$ yields a $(2 - \tfrac{\varepsilon}{2})$-approximation for \textsc{Vertex Cover}.
\end{proof}

On the positive side, we can derive the following tight bound on the gap between the minimum size of a identifying set for $\smash{\Xpath}$ and for its convex hull, the set of $s$-$t$-flows of value $1$. 

\begin{theorem}\label{thm:paths-flow-gap}
    Let $G = (V, E)$ be a DAG, let $s, t \in V$. Let~$S$ be a minimum-size identifying set for~$\smash{\Xpath}$ and let~$S'$ be a minimum-size identifying set for~$\smash{\Xflow}$.
    Then $\smash{|S'| \leq \frac{(|S| + 1)|S|}{2}}$. Moreover, for every $k \in \mathbb{N}$, there exists an instance where $|S| = k$ and the aforementioned bound holds with equality.
\end{theorem}

Before we prove \cref{thm:paths-flow-gap}, we observe that it yields the following approximation result.

\begin{corollary}
    There is a $\smash{\sqrt{|E|}}$-approximation for computing a minimum-size identifying set for~$\smash{\Xpath}$ when $G$ is DAG.
\end{corollary}

\begin{proof}
    By \cref{cor:flows}, we can compute a minimum-size identifying set for $\smash{\Xflow}$ in polynomial time.
    Note that $S'$ is identifying for $\smash{\Xpath}$ because $\smash{\Xpath \subseteq \Xflow}$ (in fact, $\smash{\Xflow = \operatorname{conv}(\Xpath)}$.
    Let~$S$ be a minimum-size identifying set for $\smash{\Xpath}$.
    We distinguish two cases.
    First, if $|S| \geq \sqrt{|E|}$, then $|S'| \leq |E| \leq \sqrt{|E|}|S|$.
    Second, if $|S| < \sqrt{|E|}$, then $|S'| \leq \frac{(|S| + 1)|S|}{2} \leq \sqrt{|E|}|S|$, where the first inequality follows from \cref{thm:paths-flow-gap}.
    Thus in both cases $|S'| \leq \sqrt{|E|} |S|$ and we obtain the desired approximation result.
\end{proof}

We now prove \cref{thm:paths-flow-gap}.

\begin{proof}[Proof of \cref{thm:paths-flow-gap}]
    Without loss of generality, we assume that for each $e \in E$, there is an~$s$-$t$-path in $E$ that contains $e$ (we can efficiently identify arcs for which this is not the case and remove them).
    We fix a topological ordering $\psi : V \rightarrow \{1, \dots, |V|\}$ of the nodes of $G$, i.e., $(v, w) \in E$ implies that $\psi(v) \leq \psi(w)$.
    Let $T \subseteq E$ be an arborescence in $G$ rooted at $s$.
    Let $S' := E \setminus T$.
    Note that the arborescence $T = E \setminus S'$ contains no undirected cycles, and therefore \cref{thm:flows} implies that~$S'$ is identifying for $X'$.
    Let
    \begin{align*}
    V_s &:= \{s\} \cup \{v \in V : \text{$v$ is the head of an arc in $S \cap T$}\}, \text{ and }\\
    V_t &:= \{t\} \cup \{v \in V : \text{$v$ is the tail of an arc in $S$}\}.
    \end{align*}
    For each arc $e = (v, w) \in S'$, let $v_s(e)$ be the last element of $V_s$ on the unique $s$-$v$-path $T[s, v]$ in $T$. Let further $v_t(e)$ be minimal with respect to $\psi$ in $\{u \in V_t : \text{there is a $w$-$u$-path in $(V, E \setminus S)$}\}$.
    We will use the following claims which we prove in \Cref{app:claim-proofs}. 

    \begin{restatable}{claim}{distinctPairs}
    \label{clm:distinct-pairs}
        Let $e, e' \in S' \setminus S$. If $v_s(e) = v_s(e')$ and $v_t(e) = v_t(e')$ then $e = e'$.
    \end{restatable}

    \begin{restatable}{claim}{numberOfPairs}
    \label{clm:number-of-pairs}
        Let $e \in S' \setminus S$. Then the following statements hold true:
        \begin{enumerate}
            \item $\psi(v_s(e)) \leq \psi(v_t(e))$\label{prop:topological}
            \item There is no $v_s(e)$-$v_t(e)$-path in $T \setminus S$.\label{prop:no-path}
        \end{enumerate}
    \end{restatable}

    For $v \in V_t$, let $\Gamma(v) := \{u \in V_S : \exists e \in S' \setminus S, v_s(e) = u, v_t(e) = v\}$.
    Note that $|S' \setminus S| = \sum_{v \in V_t} |\Gamma(v)|$ by \cref{clm:distinct-pairs}.  
    To bound $\sum_{v \in V_t} |\Gamma(v)|$, we distinguish two cases.
    
    First, consider any $v_t \in V_t$ that is the head of an arc in $S \setminus T$.
    Let $u'$ be the last node of $V_s$ on the unique $s$-$v$-path in $T$.
    By \cref{clm:number-of-pairs} statement (2), $u \notin \Gamma(v_t)$ and therefore $|\Gamma(v)| \leq |V_s| - 1 = |S \cap T|$ in this case.
    Second, we order the nodes $v \in V_t$ that are the head of some arc $e = (u, v) \in S \cap T$ according to the topological ordering of the tails $u$ of the corresponding arcs. 
    Let $v_1, \dots, v_k$ with~$k = |S \cap T|$ denote this ordering of these nodes and let $u_1, \dots, u_k$ denote the ordering of the corresponding heads. Note that in fact $V_s := \{u_1, \dots, u_k\} \cup \{s\}$.
    Let us further denote $v_{k+1} := t$ and $u_0 := s$.
    We show that $|\Gamma(v_i)| \leq i - 1$ for each $i \in [k+1]$.
    To see this, note that $u_j \notin \Gamma(v_i)$ for $j \geq i$ by \cref{clm:number-of-pairs} statement (1).
    Moreover, there is $j' < i$ such that $u_{j'}$ is the last node in $V_s$ on the unique~$s$-$v_u$-path in $T$, and hence $u_{j'} \notin \Gamma(v_i)$ by \cref{clm:number-of-pairs} statement (2).
    Thus $|\Gamma(v_i)| \leq i - 1$ for each $i \in [k+1]$.
    We conclude that
    \begin{align*}
        |S' \setminus S| & = \sum_{v \in V_t} |\Gamma(v)| \leq
        |S \setminus T| \cdot |S \cap T| + \sum_{i = 1}^{|S \cap T| + 1} (i - 1)\\
        & \leq |S \setminus T| \cdot |S \cap T| + \frac{(|S \cap T| + 1) \cdot |S \cap T|}{2}.
    \end{align*}
    Therefore,
    \begin{align*}
        |S'|  &= |S' \setminus S| + |S \setminus S'| = |S' \setminus S| + |S \setminus T|  = |S \setminus T| \cdot (|S \cap T| + 1) + \frac{(|S \cap T| + 1) \cdot |S \cap T|}{2}\\
         &= \frac{(|S \cap T| + 1) (|S| + |S \setminus T|)}{2}
         = \frac{(|S| - |S \setminus T| + 1) \cdot (|S| + |S \setminus T|)}{2} \\
         &= \frac{(|S| + 1) \cdot |S| + (1 - |S \setminus T|) \cdot |S \setminus T|}{2}
        \leq \frac{(|S| + 1) \cdot |S|}{2},
    \end{align*}
    where the final inequality follows from $(1 - |S \setminus T|) \cdot |S \setminus T| \leq 0$ because~$|S \setminus T|$ is a nonnegative integer.

    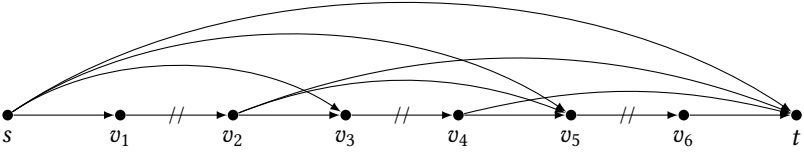
\begin{figure}[tb]
    \begin{center}
    \begin{tikzpicture}[xscale=1.5,yscale=1]
        \path node[node,label=below:{$s$}] (s) {}
            ++(1, 0) node[node,label=below:{$v_1$}] (v1) {}
            ++(1, 0) node[node,label=below:{$v_2$}] (v2) {}
            ++(1, 0) node[node,label=below:{$v_3$}] (v3) {}
            ++(1, 0) node[node,label=below:{$v_4$}] (v4) {}
            ++(1, 0) node[node,label=below:{$v_5$}] (v5) {}
            ++(1, 0) node[node,label=below:{$v_6$}] (v6) {}
            ++(1, 0) node[node,label=below:{$t$}] (t) {};
        
        \draw[-latex] (s) -- (v1);
        \draw[-latex] (s) edge[bend left=45] (v3);
        \draw[-latex] (s) edge[bend left=45] (v5);
        \draw[-latex] (s) edge[bend left=45] (t);
        
        \draw[-latex] (v1) -- node[midway, scale=0.7] {$//$} (v2);
        \draw[-latex] (v2) -- (v3);
        \draw[-latex] (v2) edge[bend left=30] (v5);
        \draw[-latex] (v2) edge[bend left=30] (t);

        \draw[-latex] (v3) -- node[midway, scale=0.7] {$//$} (v4);
        \draw[-latex] (v4) -- (v5);
        \draw[-latex] (v4) edge[bend left=20] (t);

        \draw[-latex] (v5) -- node[midway, scale=0.7] {$//$} (v6);
        \draw[-latex] (v6) -- (t);
        
    \end{tikzpicture}
    \end{center}

    \vspace*{-0.3cm}
    
    \caption{%
    \label{fig:paths-flows-tight}%
    Example showing that the bound given in \cref{thm:paths-flow-gap} is tight with $k = 3$. The marked edges $(v_1, v_2), (v_3, v_4), (v_5, v_6)$ form an identifying set for the $s$-$t$-paths $X^G_{s, t}$. Any minimum-size identifying set for the set of $s$-$t$-flows of value $1$, however, is the complement of a spanning tree in the underlying undirected graph and thus has cardinality $|E| - (|V| - 1) = k(k+1)/2 = 6$.
    }
    \end{figure}

    To show that the bound is tight, we construct the following family of example instances, one for each $k \in \mathbb{N}$.
    The set of nodes is $V = \{v_0, \dots, 2k + 1\}$ with $s = v_0$ and $t = v_{k+1}$.
    For~$i \in \{0, \dots, k\}$ we add the arcs $(v_{2i}, v_{2j + 1})$ for each $j \in \{i, \dots, k\}$. Moreover, we add the arc $e_i = (v_{2i-1}, v_{2i})$ for~$i \in \{1, \dots, k\}$.
    See \cref{fig:paths-flows-tight} for an illustration of the construction.

    Note that $S = \{e_1, \dots, e_k\}$ is an identifying set for $\smash{\Xpath}$. This is because for any $s$-$t$-path $P$ and~$i \leq j$, we have $(v_{2i}, v_{2j + 1}) \in P$ if and only if $e_i \in P$ (or $i = 0$, in which case $v_{2i} = 2$), $e_{j+1} \in P$ (or $j = k$, in which case $v_{2j + 1} = t$) and $e_{\ell} \notin P$ for all $\ell \in \{i+1, \dots, j\}$.
    Hence any $s$-$t$-path $P$ in $G$ is uniquely determined by $P \cap S$.

    Let $S'$ be a minimum-size identifying set for the set of $s$-$t$-flows.
    Note that every  arc in the constructed graph lies on an $s$-$t$-path and thus $S'$ must be the complement of a spanning tree in the underlying undirected graph of $G$.
    Thus, $|S'| = |E| - (|V| - 1)$.
    Since $|V| = 2k+2$ and
    $$|E| = \sum_{i = 0}^{k} \bigl(k - (i-1)\bigr) + k = \frac{(k+2)(k+1)}{2} + k,$$
    we conclude that
    $$|S'| = \frac{(k+2)(k+1)}{2} - (k + 1) = \frac{k(k+1)}{2} = \frac{|S|(|S| + 1)}{2},$$
    which shows that the bound given in the theorem is tight.
\end{proof}

\subsection{Controlling $s$-$t$-paths in general digraphs}
\label{sec:paths-general-digraphs}

We now turn to the case of general digraphs.
We first derive two hardness results, showing that it is hard even to verify whether a given set of arcs is identifying (contrasting \cref{thm:paths-DAG-verification}) for~$\smash{\Xpath}$ in general digraphs and that moreover, unless $\mathsf{P} = \mathsf{NP}$, the problem of finding a minimum-size identifying set in this setting does not allow for a non-trivial approximation, i.e., one that is significantly better than the trivial $|E|$-approximation achieved by simply choosing $S = E$ (contrasting \cref{thm:paths-flow-gap}).

\begin{theorem}\label{thm:path-general-digraphs-hard}
    The following problem is $\mathsf{NP}$-hard:
    Given a digraph $G = (V, E)$, two nodes $s, t \in V$, and $S \subseteq E$, decide whether $S$ is identifying for $\smash{\Xpath}$.\\
    Moreover, unless $\mathsf{P} = \mathsf{NP}$, the following problem does not admit an $O(|E|^{1-\varepsilon})$-approximation for any~$\varepsilon > 0$: 
    Given a digraph $G = (V, E)$ and $s, t \in V$, find a minimum-size identifying set for $\smash{\Xpath}$.
\end{theorem}
\begin{proof}
    We reduce from the following problem: Given a digraph $G' = (V, E')$, two nodes $s, t \in V$, and an arc $e \in E'$, does there exist an $s$-$t$-path that contains $e$?
    This problem is well-known to be $\mathsf{NP}$-hard, as it is equivalent to finding a pair of node-disjoint paths, one from $s$ to the tail of $e$ and one from the head of $e$ to $t$~\citep{fortune1980directed}.

    Let $\varepsilon > 0$.
    We construct a digraph $G = (V, E)$  where $E$ is derived from $E'$ by replacing $e = (v, w)$ with a bundle $B$ of~$\smash{|E'|^{\lceil 2/ \varepsilon \rceil} + 1}$ parallel arcs from $v$ to $w$.
    Note that if $e$ is not contained in an $s$-$t$-path in $G'$ then none of the edges in $B$ is contained in an $s$-$t$-path in $G$; hence in this case $E \setminus B$ is an identifying set for $\smash{\Xpath}$.
    Conversely, if there is an $s$-$t$-path $P$ in $G'$ containing $e$ then $P \setminus \{e\} \cup \{e'\}$ is an $s$-$t$-path in $G$ for every $e' \in B$; hence in this case, every identifying set contains at least $|B| - 1$ edges from $B$.
     
    By the above, deciding whether $S := E \setminus B$ is identifying for $\smash{\Xpath}$ is equivalent to deciding whether $G'$ has an $s$-$t$-path containing $e$, and thus the former problem is also $\mathsf{NP}$-hard.
    Likewise, any~$|E|^{1-\varepsilon}$-approximation algorithm for finding a minimum-size identifying set for $\smash{\Xpath}$ can distinguish between the case where the minimum size of such a set is at most $|E \setminus B| = |E'| - 1$, which is equivalent to $e$ not being contained in an $s$-$t$-path in $G'$, or where it is at least $$|B| - 1 = |E'|^{\lceil 2/\varepsilon \rceil} > |E'|^{\lceil 2/\varepsilon \rceil - 1} (|E'| - 1) > |E|^{1- \varepsilon} (|E'|-1),$$ which is equivalent to $e$ being contained in an $s$-$t$-path in $G'$.
    Hence, the existence of such an approximation algorithm implies $\mathsf{P} = \mathsf{NP}$.
\end{proof}

Note that both hardness results in \cref{thm:path-general-digraphs-hard} are based on the hardness of deciding whether a given arc appears on any $s$-$t$-path or not, raising the question whether this is the only source of hardness for the problem of finding small identifying sets for~$\smash{\Xpath}$ in general digraphs.
Our next theorem answers this question in the negative by establishing $\Sigma\mathsf{_2^P}$-hardness.
In particular, safe for a collapse of the polynomial hierarchy, this has the following implications:
\begin{itemize}
    \item The problem of finding small identifying sets for $\smash{\Xpath}$ cannot be formulated as an integer program of polynomial size.
    \item Even after a hypothetical preprocessing step that would remove all arcs not appearing on any $s$-$t$-path, the problem cannot be solved in polynomial time.\footnote{Note that such a preprocessing step can be carried out by a non-adaptive $\mathsf{NP}$-oracle and hence a polynomial-time algorithm using such an oracle would imply $\Sigma\mathsf{_2^P} = \mathsf{\theta_2^P}$.}
\end{itemize}

\begin{theorem}\label{thm:paths-sigma-2-p} \;
    The following problem is $\Sigma\mathsf{^P_2}$-complete: \;
    Given a digraph $G = (V, E)$, two nodes~$s, t \in V$, and $k \in \mathbb{N}$, decide whether there is an identifying set $S$ for $\smash{\Xpath}$ with $|S| \leq k$.
\end{theorem}

\begin{proof}[Proof (sketch)]
    To establish hardness, we make use of an auxiliary problem, called \textsc{Forbidden Pairs $s$-$t$-Path Interdiction}, which asks for a set of edges that intersect any $s$-$t$-path except those containing certain forbidden pairs of arcs.
    We show that this problem is $\Sigma\mathsf{^P_2}$-hard via reduction from \textsc{2-Quantified 3-DNF-Satisfiability}.
    Then, to reduce  \textsc{Forbidden Pairs $s$-$t$-Path Interdiction} to the problem of finding a small identifying set for $\Xpath$, we make use of a gadget introduced by \citet{fortune1980directed} for showing $\mathsf{NP}$-hardness of \textsc{Disjoint Paths} with two commodities.
    We show that by appropriately embedding these gadgets into three distinct copies of the graph and combining them with a new gadget that forces certain arcs to be contained in any identifying set, we can employ them in the single-commodity context of our identifying set problem.
    This yields a sequence of four intermediate reductions that ultimately produce a digraph in which the identifying sets for $\smash{\Xpath}$ correspond exactly to the arc sets interdicting all valid paths from the \textsc{Forbidden Pairs $s$-$t$-Path Interdiction} instance.
    The complete proof can be found in \cref{app:paths-sigma-2-p}.
\end{proof}

\section{Controlling Matroids and Polymatroids}
\label{sec:matroids}

In this section, we discuss controlling sets for the case that $X$ is described by a polymatroid or matroid.

Let~$B(f)$ be a polymatroid with polymatroid rank function~$f$.
The set~$B(f)$ is clearly convex, so \Cref{thm:convex} can be applied to obtain a characterization of the controlling and identifying sets for~$B(f)$.
To this end, some polymatroid terminology is needed.
A polymatroid $(E,f)$ is \emph{connected} if there is no non-empty proper subset $T \subsetneq E$ such that $f(T) + f(E \setminus T) = f(E)$.
Every polymatroid~$(E,f)$ has a unique partition into connected polymatroids, i.e., there is a unique partition $\{E_1,\dots,E_k\}$ of $E$ such that $(E_i,f)$ is connected for all~$i \in [k]$ and $f(T) = \sum_{i \in [k]} f(E_i \cap T)$ for all $T \subseteq E$, and we then have $\dim B(f) = |E| - k$; see \citet[Theorem 3.38 \& Corollary 3.40]{fujishige1991submodular}.
Thus, \Cref{thm:convex} implies that $S$ has to contain at least $|E|-k$ elements. 
The following yields a complete characterization of the controlling sets of a polymatroid in terms of its connected components.

\begin{theorem}
\label{thm:polymatroid}
Let $X = B(f)$ for some polymatroid rank function~$f$ and let~$\C$ be the set of convex functions.
Let $\{E_1,\dots,E_k\}$ be the partition of $E$ induced by the connected components of $X$.
Then, the following are equivalent:
\begin{enumerate}[label=(\roman*.)]
    \item $S$ is a controlling set for $(X,\C)$. \label{it:polymatroid-controlling}
    \item $S$ is an identifying set for $X$. \label{it:polymatroid-identifying}
    \item $|S \cap E_i| \geq |E_i| - 1$ for all $i \in [k]$.\label{it:polymatroid-partition}
\end{enumerate}
\end{theorem}

\begin{proof}
The equivalence between \ref{it:polymatroid-controlling} and \ref{it:polymatroid-identifying} follows from \Cref{thm:convex}.
We proceed to show the implication~\ref{it:polymatroid-identifying} $\Rightarrow$ \ref{it:polymatroid-partition}
For a contradiction, suppose that $S$ is identifying but there is $i \in [k]$ such that $S \cap E_i \leq |E_i| - 2$.
Using \citet[Theorem~3.36]{fujishige1991submodular}, there is $x \in X$ such that $x(T) < f(T)$ for all $T \subsetneq E_i$ with $T \neq \emptyset$.
Let $e, e' \in E_i \setminus S$ with $e \neq e'$.
We claim that $x_{e} > 0$. Indeed, using that $x(T) < f(T)$ for all $T \subsetneq E_i$ with $T \neq \emptyset$, we have in particular that 
\begin{align*}
x_e = x(E_i) - x(E_i \setminus \{e\}) > f(E_i) - f(E_i \setminus \{e\}) \geq 0, 
\end{align*}
where for the strict inequality we used that $f(E_i) = x(E_i)$ for every ~$x \in X$.
This then implies that there is $\epsilon > 0$ such that $y = x + \epsilon(\chi_{e'} - \chi_{e}) \in X$.
By construction, $x_S = y_S$, a contradiction to the assumption that~$S$ is identifying.

To show \ref{it:polymatroid-partition} $\Rightarrow$ \ref{it:polymatroid-identifying} suppose that $|S \cap E_i| \geq |E_i| -1$ for all $i \in [k]$. Then for each $x,y \in X$ with~$x_S = y_S$, we have~$x(E_i) = y(E_i)$ for all~$i \in [k]$ and hence also~$x = y$ so that $S$ is identifying.
\end{proof}

\Cref{thm:polymatroid} also gives rise to an efficient algorithm optimizing a linear objective over the set of controlling sets for $X = B(f)$, when $f$ is given by a value oracle, that given $U \subseteq E$ returns $f(U)$.

\begin{corollary}
\label{cor:polymatroid}
There is an (oracle) polynomial-time algorithm solving
\begin{align*}
\min \bigl\{ w^\top \chi_S : S \text{ is identifying for $X$} \bigr\}, 
\end{align*}
when $w \in \R^E$ and $X = B(f)$ is a polymatroid given by a value oracle for $f$.
\end{corollary}

\begin{proof}
By \Cref{thm:polymatroid}~\ref{it:polymatroid-partition}, optimizing over $S$ is equivalent to optimizing over the dual of a partition matroid which can be done with the greedy algorithm that simply removes the most expensive element from each set $E_i$, $i \in [k]$. The only thing left to argue is that we can compute the partition $\{E_1,\dots,E_k\}$ in polynomial time with access to a value oracle for $f$.

Consider the \emph{dependence function} of the polymatroid $\dep : B(f) \times E \to 2^E$ defined as
\begin{align*}
\dep(x,e) = \bigl\{e' \in E : \text{ there is $\epsilon > 0$ with $x + \epsilon(\chi_e - \chi_{e'}) \in B_f$} \bigr\}. 
\end{align*}
Intuitively, $\dep(x,e)$ contains those elements~$e'$ from which a sufficiently small shift from $e'$ to $e$ does not leave the feasible region.
The dependence function can be computed in polynomial time when having access to a value oracle for $f$; see, e.g., the discussion by \citet[\S~3.3c]{fujishige1991submodular}.
From the dependence function, define a directed graph $G(x) = (E, A(x))$ with
$A(x) = \bigl\{(e,e') : e' \in \dep(x,e)\bigr\}$.
The partition $\{E_1,\dots,E_k\}$ is the equal to the connected components of $G(x)$ \citep[Lemma 3.41]{fujishige1991submodular}.
In summary, the partition $\{E_1,\dots,E_k\}$ can be computed in polynomial time and the result follows.
\end{proof}

We proceed to give a similar characterization for matroids.

\begin{theorem}
\label{thm:matroid}
    Let $\mathcal{M} = (E,\mathcal{I})$ be a matroid and $X$ be the set of incidence vectors of bases of $\mathcal{M}$ and let~$\mathcal{C}$ be the class of non-decreasing cost functions. Then, the following are equivalent:
\begin{enumerate}[label=(\roman*.)]
\item $S$ is a controlling set for $(X,\mathcal{C})$.\label{it:matroid-controlling}
\item $S$ is an identifying set for $X$.\label{it:matroid-identifying}
\item $|S \cap C| \geq |C| - 1$ for all circuits $C$ of $\mathcal{M}$. \label{it:matroid-circuit}
\end{enumerate}    
\end{theorem}

\begin{proof}
\ref{it:matroid-controlling} $\Rightarrow$ \ref{it:matroid-identifying} follows from \Cref{lem:ctrl=ident}.
\ref{it:matroid-identifying} $\Rightarrow$ \ref{it:matroid-controlling} follows from  \Cref{thm:identifying-is-controlling-discrete}. We proceed to show $\neg$\ref{it:matroid-circuit} $\Rightarrow$ $\neg$\ref{it:matroid-identifying}.

Let $S$ be such that there is a circuit~$C$ with $|S \cap C| \leq |C|-2$. We show that then $S$ is not identifying. Let $e,f \in C \setminus S$ with $e \neq f$ be arbitrary.
As $C$ is minimal dependent, we have that~$C \setminus \{f\}$ is independent and, in particular, there is a basis~$B$ such that~$C \subseteq B \cup \{f\}$.
A basic fact for matroids is that adding an element to a set closes at most one circuit, so when adding $f$ to $B$, the unique circuit closed must be $C$.
Since $e \in C$, we have that $C' = (B \cup \{f\}) \setminus \{e\}$ is a basis as well.
By construction, $B$ and $B'$ cannot be distinguished with $S$, completing the proof for this direction.

We proceed to show \ref{it:matroid-circuit} $\Rightarrow$ \ref{it:matroid-identifying}. Let $\mathcal{K}$ denote the circuits of $\mathcal{M}$.
For a contradiction, assume that~$|S \cap C| \geq |C| -1$ for all circuits $C \in \mathcal{K}$ but there are bases $B,B'$ with $B \neq B$ such that~$B \cap S = B' \cap S$.
Let $e \in B \setminus B'$ be arbitrary.
If~$e \in S$, we have $e \in B \cap S$ but $e \notin B' \cap S$, a contradiction.
So we can assume that $e \notin S$.
Then,~$B' \cup \{e\}$ closes a fundamental circuit~$C \in \mathcal{K}$.
Since~$e \in C$ and $e \notin S$ our assumption on $S$ implies that $C \setminus \{e\} \subseteq S$.
We further have~$C \setminus \{e\} \subseteq B'$.
As~$B$ and~$B'$ coincide on~$S$, this yields that $C \setminus \{e\} \subseteq B$ and further, since $e \in B$, that $C \subseteq B$. This is a contradiction since $B$ cannot contain a circuit.
\end{proof}

\Cref{thm:matroid} also gives rise to an efficient algorithm optimizing a linear objective over the set of controlling sets when $X$ is the set of incidence vectors of bases of a matroid $\mathcal{M}$ and $\mathcal{M}$ is given by an independence oracle.

\begin{corollary}
\label{cor:matroid}
There is an (oracle-)polyomial-time algorithm solving
\begin{align*}
\min \bigl\{ w^\top \chi_S : S \text{ is identifying for } X\bigr\}, 
\end{align*}
where $w \in \mathbb{R}^E$ and $X$ is the set of incidence vectors of a matroid $\mathcal{M}$ given by an independence oracle.
\end{corollary}

\begin{proof}
For a matroid $\mathcal{M}$, one can define an equivalence relation $\sim_{\mathcal{M}}$ by $e \sim_{\mathcal{M}} f$ if either $e=f$, or there is a circuit $C$ with $\{e,f\} \subseteq C$. For a proof that this indeed is an equivalence relation, see, e.g., \citep[Proposition 4.1.2]{oxley2011matroid}.
The equivalence classes of this equivalence relation are also called the \emph{connected components} of the matroid.

It is not hard to convince ourselves that the connected components of a matroid can be computed in oracle-polynomial time.
For this, take an arbitrary basis $B$ of the matroid~$\mathcal{M}$ and define the \emph{fundamental graph} of $\mathcal{M}$ with respect to $B$ as follows:
The graph is bipartite and the nodes on the left side of the partition are the elements in $B$ and the nodes on the right side of the partition are the elements in $E \setminus B$. There is an edge $\{i,j\}$ with $i \in B$ and $j \in E \setminus B$ if and only if $i$ is contained in the fundamental cycle closed by $j$. A consequence of \citet[Proposition~4.3.2]{oxley2011matroid} is that the connected components of this graph correspond to the connected components of the matroid. For the computation of the graph, we only need to be able to compute fundamental cycles for a given basis and a given element outside the basis which clearly can be done with polynomial many calls to an independence oracle.

Having computed the connected components, \Cref{thm:matroid} implies that $S$ needs to contain all but one element from each connected component.
Thus, the optimization can be done by a simple greedy algorithm.
\end{proof}

\section{Controlling $X$ When the List of Solution Is Given Explicitly}
\label{sec:explicit}

In the preceding sections, we considered several cases in which $X$ were given implicitly and potentially of exponential size in the encoding length of the input (this is, e.g., the case when $X$ is the set of $s$-$t$-paths of a given digraph).
In this section, we consider the complexity of finding a minimal identifying set when $X \subseteq \{0, 1\}^E$ is given explicitly in the input.
It turns out that this setting is closely related to the \textsc{Set Cover} problem, which, given a finite universe $U$ and a set system~$\mathcal{S} \subseteq 2^U$ asks for minimum-size subset $\mathcal{S}' \subseteq \mathcal{S}$ such that $\bigcup_{T \in \mathcal{S}'} T = U$.
It is well known that \textsc{Set Cover} allows for a $\ln |U|$-approximation (even when looking for a min-weight instead of min-size subset of $\mathcal{S}$) via a simple greedy algorithm and that this is best possible, unless $\mathsf{P} = \mathsf{NP}$~\cite{feige1998threshold}.
We first show that this approximation result carries over:

\begin{theorem}
    There is an $2 \ln |X|$-approximation algorithm for the following problem:
    Given~$X \subseteq \{0, 1\}^E$ and weights $w \in \mathbb{R}^E_+$, find an identifying set $S$ for $X$ minimizing $\sum_{e \in S} w_e$.
\end{theorem}
\begin{proof}
    Let $U := \binom{X}{2}$.
    For $e \in E$, let $\smash{T_e := \bigl\{\{x, x'\} \in U : x_e \neq x'_e\bigr\}}$.
    Note that $S \subseteq E$ is identifying for $X$ if and only if $\bigcup_{e \in S} T_e = U$, i.e., if the sets $\{T_e : e \in S\}$ form a set cover for $U$.
    Hence, applying the greedy algorithm for \textsc{Set Cover}, to the instance with ground set $U$, set system $\mathcal{S} := \{T_e : e \in E\}$, and weight $w_e$ for each set $T_e$, yields an identifying set $S$ whose weight $\sum_{e \in S} w_e$ exceeds the weight of a min-weight identifying set by at most a factor of $\ln |U| \leq \ln |X|^2 = 2 \ln |X|$.
\end{proof}

We also obtain the following hardness-of-approximation result, even for the unweighted version of the problem.

\begin{theorem}
    Unless $\mathsf{P} = \mathsf{NP}$, there is no $\frac{1}{4}\ln |X|$-approximation for the following problem: Given~$X \subseteq \{0, 1\}^E$, find an identifying set $S$ for $X$ minimizing $|S|$.
\end{theorem}
\begin{proof}
    Given an instance of \textsc{Set Cover} with ground set $U$ and set system $\mathcal{S} \subseteq 2^U$, we construct the following set $X$.
    First, without loss of generality, we assume that $U := \{0, \dots, n-1\}$ for $n := |U|$.
    We choose $k := \lceil \log_2 n \rceil$ and $\ell := 2^k$.
    We let $$E := \bigl\{e_{T, i} : T \in \mathcal{S}, i \in [\ell]\bigr\} \cup \bigl\{e'_j : j \in \{0, \dots, k - 1\}\bigr\}.$$
    For $h \in \{0, \dots, \ell - 1\}$, let $\bar{x}^h \in \{0, 1\}^E$ denote the unique vector with $\bar{x}^h_{e_{T, i}} = 0$ for all $T \in \mathcal{S}$ and all~$i~\in [\ell]$ and~$\sum_{j = 0}^{k-1} 2^j \bar{x}^h_{e'_j} = h$, i.e., the $e'_i$-components of $\bar{x}^h$ correspond to the binary encoding of~$h$.
    For $u \in U$ define $\smash{x^{u,i} \in \{0, 1\}^E}$ by 
    \begin{align*}
        x^{u,i}_{e} := 
        \begin{cases}
            1 & \text{ if } e = e_{T, i} \text{ and } u \in T,\\
            \bar{x}^u_e & \text{ if } e = e'_j \text{ for some } j \in \{0, \dots, k-1\},\\ 
            0 & \text{otherwise.}
        \end{cases}
    \end{align*}
    In particular, the components $e'_0, \dots, e'_{k-1}$ of $x^{u,i}$ contain the binary representation of the number $u$ and the components $e_{T,i}$ for $T \in \mathcal{S}$ contain the indicator $e \in T$.
    We let $$X := \bigl\{x^{u,i} : u \in U, i \in \{0, \dots, \ell-1\}\bigr\} \cup \bigl\{\bar{x}^h : h \in \{0, \dots, \ell-1\}\bigr\}.$$
    We make the following observation.
    \begin{claim}\label{clm:set-cover}
        A set $S \subseteq E$ is identifying for $X$ if and only if $e'_i \in S$ for all $i \in \{0, \dots, \ell - 1\}$ and~$\bigcup_{T \in \mathcal{S} : e_{T,i} \in S} T = U$ for all $i \in [\ell]$.
    \end{claim}
    \begin{proof}\renewcommand{\qedsymbol}{$\blacklozenge$}
        Note that for every $j \in \{0, \dots, k-1\}$ there are $h, h' \in \{0, \dots, \ell - 1\}$ with $h \neq h'$ such that~$\smash{\bar{x}^h_e = \bar{x}^{h'}_e}$ for all $e \in E \setminus \{e'_j\}$.
        Hence, any identifying set for $X$ must contain $\bar{e}_j$ for each~$j \in \{0, \dots, k-1\}$.
        Note further that, for each $u \in U$ and $i \in [\ell]$, it holds that $\smash{x^{u,i}_e = \bar{x}^h_e}$ for all~$e \in E \setminus \{e_{T_i} : u \in T\}$.
        Hence, any identifying set must contain some $e_{T,i}$ with $u \in T$ for each~$i \in [\ell]$.

        Conversely, consider any $x, x' \in X$ with $x \neq x'$.
        We show that if $S$ is of the form described in the statement of the claim, then there is $e \in S$ with $x_e \neq x_e$.
        If~$x_{e'_j} \neq x'_{e'_j}$ for some $j$ then we are done.
        If~$x_{e'_j} = x'_{e'_j}$ for some $j$, then there is $u \in U$ such that~$x, x' \in \{\bar{x}^u\} \cup \{x^{u,i} : i \in [\ell]\}$.
        Without loss of generality, assume that $x = x^{u,i}$ for some $i$.
        Then there is $T \in \mathcal{S}$ such that $u \in T$ and $e_{T,i} \in S$.
        Hence~$\smash{x_{e_{T,i}} = 1 \neq 0 = x'_{e_{T,i}}}$, because either $\smash{x' = \bar{x}^u}$ or $\smash{x' = x^{u,i'}}$ for some $i' \neq i$.
    \end{proof}
    Assume by contradiction that there is a $\frac{1}{4}\ln |X|$-approximation algorithm for finding a minimum-size identifying set.
    By \cref{clm:set-cover}, there is an identifying set for $X$ of size $k + \ell \cdot m^*$ where $m^*$ is the minimum size of a set cover.
    Hence the approximation algorithm will return an identifying set $S$ of size at most $\frac{1}{4}\ln |X|(k + \ell \cdot m)$.
    Also by \cref{clm:set-cover}, the set $\mathcal{S}'_i := \{T \in \mathcal{S} : e_{T,i} \in S\}$ is a set cover for each $i \in [\ell]$. 
    By the pigeonhole principle, there is some $i$ such that 
    \begin{multline*}
        |\mathcal{S}'_i|  \leq \frac{1}{4}\ln |X|\biggl(\frac{k}{\ell} +  m^* \biggr)= \frac{1}{4} \ln \bigl(\ell(|U| + 1)\bigr) \biggl(\frac{k}{\ell} +  m^* \biggr) 
         = \frac{1}{4} (\ln \ell + \ln (|U| + 1))) \biggl(\frac{k}{\ell} +  m^*\biggr) \\ = \frac{1}{4}  \cdot (\ln 2 \cdot \lceil \log_2 |U| \rceil + \ln (|U| + 1)) \biggl(\frac{k}{\ell} + m^* \biggr) \leq \frac{1}{4} \cdot 2 \cdot (1 + \ln |U|) (m^* + 1),
    \end{multline*}
    which is less than $\frac{18}{25}\ln |U| \cdot m^*$ for $|U| > 32$ and $m^* > 5$.
    Thus, a $\frac{1}{4}\ln |X|$-approximation algorithm for finding a minimum-size identifying sets implies a better-than-$\ln |U|$-approximation for \textsc{Set Cover} and thus $\mathsf{P} = \mathsf{NP}$.
\end{proof}

\newpage
\section*{Appendix}

\appendix

\section{Claims used in the proof of Theorem~\ref{thm:paths-hardness-DAGs}}
\label{app:hardness-DAGs}

\identifyingSetVertexCover*

\begin{proof}\renewcommand{\qedsymbol}{$\blacklozenge$}
        Beginning with $S' := S$, we repeat the following as long as there is some $e \in \bar{E}$, $v \in e$, and~$i \in [\ell]$ with $(u_e, v_i) \in S'$.
        We distinguish three cases.
        \begin{itemize}
            \item If $(s, u_e) \in S'$, then we replace $(u_e, v_i)$ by $(v_i, t)$ in $S'$. Note that $S'$ remains identifying for~$\smash{\Xpath}$ because the unique $s$-$t$-path containing $(u_e, v_i)$ is the unique $s$-$t$-path containing both $(s, u_e)$ and $(v_i, t)$.
            \item If $(s, u_e) \notin S'$ but $(s, u_{e'}) \in S'$ for all $e' \in \delta_{\bar{G}}(e') \setminus \{e\}$, then we replace $(u_e, v_i)$ by $(v_i, t)$ in $S'$. Note that $S'$ remains identifying for $\Xpath$ because the unique $s$-$t$-path containing $(u_e, v_i)$ is the unique $s$-$t$-path containing $(v_i, t)$ and not containing any of the edges $(s, u_{e'})$ for $e' \in \delta_{\bar{G}}(e') \setminus \{e\}$.
            \item If $(s, u_e) \notin S'$ and there is $e' \in \delta_{\bar{G}}(e') \setminus \{e\}$ with $(s, u_{e'}) \notin S'$, then $(u_e, v_i), (u_{e'}, v_i) \in S'$ for all $i \in \{1, \dots, k\}$, as otherwise the paths $s$-$u_e$-$v_i$-$t$ and $s$-$u_{e'}$-$v_i$-$t$ are indistinguishable under $S'$ (which is identifying by induction).
            We replace $\{(u_e, v_i), (u_{e'}, v_i) : i \in [\ell]\}$ by $\{(s, u_e), (s, u_{e'})\} \cup \{(v_i, t) : i \in [\ell]$ in $S'$.
            Again, the set $S'$ remains identifying for $\Xpath$ after this transformation because the unique $s$-$t$-path containing $(u_e, v_i)$ is the unique $s$-$t$-path containing both $(s, u_e)$ and $(v_i, t)$, and likewise, the unique $s$-$t$-path containing $(u_{e'}, v_i)$ is the unique $s$-$t$-path containing both $(s, u_{e'})$ and $(v_i, t)$.
            Moreover, we replaced $2\ell$ arcs by $2 + \ell$ arcs, thus not increasing the size of $S'$.
        \end{itemize}
        Thus, we obtain an identifying set $S'$ with $|S'| \leq |S|$ and $S \cap E' = \emptyset$.
        Now assume there is $i \in [\ell]$ such that $U_i$ is not a vertex cover in $\bar{G}$.
        Then there must be $e = \{v, w\} \in \bar{E}$ such that $v, w \notin U_i$, i.e., $(v_i, t), (w_i, t) \notin S$.
        But then the two paths $s$-$u_e$-$v_i$-$t$ and $s$-$u_e$-$w_i$-$t$ are indistinguishable under $S'$ as none of the arcs $(u_e, v_i)$, $(v_i, t)$, $(u_e, w_i)$, $(w_i, t)$ is contained in $S'$.
        This is a contradiction as $S'$ is identifying for $\smash{\Xpath}$.
        Thus each $U_i$ must be a vertex cover in $\bar{G}$.
    \end{proof}

\section{Claims used in the proof of Theorem~\ref{thm:paths-flow-gap}}
\label{app:claim-proofs}

\distinctPairs*

\begin{proof}\renewcommand{\qedsymbol}{$\blacklozenge$}
        By contradiction, assume there are $e = (v, w), e' = (v', w') \in S' \setminus S$ with $v_s(e) = v_s(e')$ and~$v_t(e) = v_t(e')$ but $e \neq e'$.
        Let $Q$ be an arbitrary $w$-$v_t(e)$-path in $(V, E \setminus S)$ and let $Q'$ be an arbitrary $w'$-$v_t(e')$-path in $(V, E \setminus S)$.
        Let $R$ be an arbitrary $v_t(e)$-$t$-path.
        Let $P$ be the $s$-$t$-path resulting from the concatenation of $T[s, v]$, $e = (v, w)$, $Q$, and $R$.
        Let $P'$ be the $s$-$t$-path resulting from the concatenation of $T[s, v']$, $e' = (v', w')$, $Q'$, and $R$.
        Note that, by construction, $P[s, v_s(e)] = P'[s, v_s(e')]$ and $P[v_t(e), t] = P'[v_t(e'), t]$ for the respective subpaths from $s$ to $v_s(e)$ and from $v_t(e)$ to $t$.
        Moreover, $P[v_s(e), v_t(e)] \cap S = \emptyset =
        P'[v_s(e'), v_t(e')]$ by construction. 
        Therefore $P \cap S = P' \cap S$ and hence $P = P'$, because $S$ is identifying for $\Xpath$.
        Without loss of generality, assume that $P$ traverses $e$ before~$e'$.
        Then $e \in T[s, v']$, a contradiction to $e \in S' = E \setminus T$.
    \end{proof}

\numberOfPairs*
    
    \begin{proof}\renewcommand{\qedsymbol}{$\blacklozenge$}
        For \eqref{prop:topological}, simply note that $G$, by definition of $v_s(e)$ and $v_t(e)$, contains a $v_s(e)$-$v_t(t)$-path and hence $\psi(v_s(e)) \leq \psi(v_t(e))$ by consistency of the topological ordering.
        For \eqref{prop:no-path}, consider any $e = (v, w) \in S' \setminus S$ and construct the following two paths $P$ and $P'$ in $G$:  
        Let $Q$ be an arbitrary $w$-$v_t(e)$-path in $(V, E \setminus S)$ and let $R$ be an arbitrary $v_t(e)$-$t$-path in $G$.
        Let $P$ be the $s$-$t$-path resulting from the concatenation of $T[s, v]$, $e$, $Q$, and $R$.
        Let $P'$ be the $s$-$t$-path resulting from the concatenation of $T[s, v_t(e)]$ and $R$.
        Note that by construction $P \cap S = P' \cap S'$, as $T[v_s(e), v_t(e)] \cap S = \emptyset$. Furthermore $P \neq P'$, because $e \in S' = E \setminus T$ and thus $e \in P$ but $e \notin P'$.
        This is contradiction to $S$ being identifying for $\Xpath$.
    \end{proof}

\section{Proof of Theorem~\ref{thm:paths-sigma-2-p}}
\label{app:paths-sigma-2-p}

We first establish hardness for the following auxiliary problem, which we call \textsc{Forbidden Pairs $s$-$t$-Path Interdiction}:
Given a directed graph $D = (V, E)$ with two vertices $s, t \in V$, a set $\mathcal{F} \subseteq 2^E$ with $|F| = 2$ for all $F \in \mathcal{F}$, and $k \in \mathbb{N}$, find a set $S \subseteq E$ with $|S| \leq k$ such that $P \cap S \neq \emptyset$ for all $P \in \mathcal{P}_{\mathcal{F}}$, where $$\mathcal{P}_{\mathcal{F}} := \bigl\{P \subseteq E : P \text{ is an $s$-$t$-path in $D$}, F \setminus P \neq \emptyset \text{ for all } F \in \mathcal{F}\bigr\}.$$
To see that this problem is $\Sigma\mathsf{^P_2}$-complete, we reduce from \textsc{2-Quantified 3-DNF-Satisfiability}.

\begin{theorem}\label{thm:sigma-2-p-forbidden-pairs-interdiction}
    \textsc{Forbidden Pairs $s$-$t$ Path Interdiction} is $\Sigma\mathsf{^P_2}$-complete.
    This holds even when restricted to instances where $D$ is a directed acyclic graph and for every $e \in E$ there is $P \in \mathcal{P}_{\mathcal{F}}$ with $e \in P$.
\end{theorem}

\begin{proof}
    We start from an instance of \textsc{2-Quantified 3-DNF-Satisfiability}, i.e., we are given a Boolean formula $\phi$ in disjunctive normal form over variables $x_1, \dots, x_n$ and $y_1, \dots, y_m$ with the goal to decide whether there is an assignment to $x_1, \dots, x_n$ such that $\phi(x_1, \dots, x_n, y_1, \dots, y_m)$ is true for all assignments of the variables $y_1, \dots, y_m$.
    Let $C_1, \dots, C_{\ell}$ denote the conjunctive clauses of~$\phi$, i.e.,~$\phi = C_1 \vee \dots \vee C_{\ell}$.
    For convenience of notation, we identify each clause with the set of its literals.
    For example, if $C_1 = x_1 \wedge \neg x_2 \wedge x_3$, then $x_1, \neg x_2, x_3 \in C_1$.
    
    We construct the following instance of \textsc{Forbidden Pairs $s$-$t$-Path Interdiction}.
    The digraph~$D$ has the nodes 
    \begin{itemize}
        \item $s$, $t$,
        \item $u_i, u_{i,\text{true}}, u'_{i,\text{true}}, u_{i,\text{false}}, u'_{i,\text{false}}$ for each $i \in [n]$,
        \item $v_i, v_{i,\text{true}}, v_{i,\text{false}}$ for each $i \in [m]$,
        \item $w_j$ and $w_{j,z}$ for each $j \in [\ell]$ and each literal $z \in C_j$.
    \end{itemize}
    We identify $u_{n+1} := v_1$, $v_{m+1} := w_1$, and $w_{\ell+1} := t$.
    The digraph has the following arcs:
    \begin{itemize}
        \item the arc $(s, u_1)$,
        \item for each $i \in \{1, \dots, n\}$ the arcs
        \begin{itemize}
            \item $(u_i, u_{i,\text{true}})$, $(u_{i,\text{true}}, u'_{i,\text{true}})$, and $(u'_{i,\text{true}}, u_{i+1})$,
            \item $(u_i, u_{i,\text{false}})$, $(u_{i,\text{false}}, u'_{i,\text{false}})$, and $(u'_{i,\text{false}}, u_{i+1})$,
            \item $(s, u_{i,\text{true}})$, $(u'_{i,\text{true}}, u_{i,\text{false}})$, and $(u'_{i,\text{false}}, t)$,
        \end{itemize}
        \item for each $i \in \{1, \dots, m\}$ the arcs
        \begin{itemize}
            \item $(v_i, v_{i,\text{true}})$ and $(v_{i,\text{true}}, v_{i+1})$,
            \item $(v_i, v_{i,\text{false}})$ and $(v_{i,\text{false}}, v_{i+1})$,
        \end{itemize}
        \item for each $j \in [\ell]$ and each literal $z \in C_j$ the arcs $(w_j, w_{j,z})$ and $(w_{j,z}, w_{j+1})$.
    \end{itemize}
    The forbidden pairs are
    \begin{itemize}
        \item for each $i \in \{1, \dots, n\}$:
        \begin{itemize}
            \item $\{(s, u_1), (u'_{i,\text{false}}, t)\}$,
            \item $\{(s, u_{i,\text{true}}), (u'_{i,\text{true}}, u_{i+1})\}$,
            \item $\{(s, u_{i,\text{true}}), (u'_{i,\text{false}}, u_{i+1})\}$ ,
        \end{itemize}
        \item and for each $j \in [\ell]$ and $z \in C_j$:
        \begin{itemize}
            \item $\{(u_i, u_{i,\text{false}}), (w_j, w_{j,\neg x_i})\}$ if $z = \neg x_i$,
            \item $\{(u_i, u_{i,\text{true}}), (w_j, w_{j, x_i})\}$ if $z = x_i$,
            \item $\{(v_i, v_{i,\text{false}}), (w_j, w_{j,\neg y_i})\}$ if $z = \neg y_i$,
            \item $\{(v_i, v_{i,\text{true}}), (w_j, w_{j, y_i})\}$ if $z = y_i$.
        \end{itemize}
    \end{itemize}

    We show that there is $S \subseteq E$ with $|S| \leq n$ and $P \cap S \neq \emptyset$ for all $P \in \mathcal{P}_{\mathcal{F}}$ if and only if there is a an assignment of the $x$-variables that such that $\phi(x, y)$ is true for any assignment of the $y$-variables.

    We start by making an observation on the structure of paths in $\mathcal{P}_{\mathcal{F}}$.
    Consider any $P \in \mathcal{P}_{\mathcal{F}}$.
    
    Note that if $(s, u_{i,\text{true}}) \in P$ for some $i \in [n]$ then $P$ must correspond to the path $Q_i$ corresponding to the node sequence $s$-$u_{i,\text{true}}$-$u'_{i,\text{true}}$-$u_{i,\text{false}}$-$u'_{i,\text{false}}$-$t$. 
    This is because the forbidden pairs $\{(s, u_{i,\text{true}}), (u'_{i,\text{true}}, u_{i+1})\}$ and $\{(s, u_{i,\text{true}}), (u'_{i,\text{false}}, u_{i+1})\}$ force the path to use $(u'_{i,\text{false}}, t)$ to reach $t$.
    
    Consider any $j \in [\ell]$ and $z \in C_j$ such that $(w_j, w_{j,z}) \in P$.
    \begin{itemize}
    \item If $z = \neg x_i$ for $i \in [n]$, then forbidden pair $\{(u_i, u_{i,\text{false}}), (w_j, w_{j,\neg x_i})\}$ forces $(u_{i,\text{true}}, u'_{i,\text{true}}) \in P$.
    \item If $z = x_i$ for $i \in [n]$, then forbidden pair $\{(u_i, u_{i,\text{true}}), (w_j, w_{j, x_i})\}$ forces $(u_{i,\text{false}}, u'_{i,\text{false}}) \in P$.
    \item If $z = \neg y_i$ for $i \in [m]$, then forbidden pair $\{(v_i, v_{i,\text{true}}), (w_j, w_{j, \neg y_i})\}$ forces $(v_i, v_{i,\text{false}}) \in P$.
    \item If $z = y_i$ for $i \in [m]$, then forbidden pair $\{(v_i, v_{i,\text{true}}), (w_j, w_{j, y_i})\}$ forces $(v_i, v_{i,\text{false}}) \in P$.
    \end{itemize}
    Consider the assignment to the variables corresponding to the choices of the path $P$ (i.e., $x_i = \text{true}$ if $(u_{i,\text{true}}, u'_{i,\text{true}}) \in P$ etc.).
    Then, by the above, for each $j \in [\ell]$, the literal of the clause $C_j$ is not fulfilled by this assignment, and therefore also the entire formula is not fulfilled.
    Thus every path in $\mathcal{P}_{\mathcal{F}}$ is either of the form $Q_i$ (with node sequence $s$-$u_{i,\text{true}}$-$u'_{i,\text{true}}$-$u_{i,\text{false}}$-$u'_{i,\text{false}}$-$t$) or corresponds to an assignment of $x$ and $y$ variables that does not fulfill $\phi$.
    
    Now assume there is an assignment of the $x$-variables that such that $\phi(x, y)$ is true for any assignment of the $y$-variables.
    Let $T \subseteq [n]$ denote the set of indices of those variables set to true in the assignment.
    We let $S := \{(u_{i,\text{false}}, u'_{i,\text{false}}) : i \in T\} \cup \{(u_{i,\text{true}}, u'_{i,\text{true}}) : i \in [n] \setminus T\}$.

    Consider any $P \in \mathcal{P}_{\mathcal{F}}$.
    If $P = Q_i$ for some $i \in [n]$ then $P$ contains both $(u_{i,\text{true}}, u'_{i,\text{true}})$ and $(u_{i,\text{false}}, u'_{i,\text{false}})$ and therefore $P \cap S \neq \emptyset$. 
    If $P$ corresponds to an assignment not satisfying the clauses, then there must be at least one $i \in [n]$ for which $i \in T$ and $(u_{i,\text{false}}, u'_{i,\text{false}}) \in P$ or $i \notin T$ and $(u_{i,\text{true}}, u'_{i,\text{true}}) \in P$; otherwise the $x$-part of the assignment given by $P$ corresponds to the fulfilling assignment described by $T$, a contradiction.
    Hence $P \cap S \neq \emptyset$ also in this case.

    Finally, consider any set $S \subseteq E$ with $|S| \leq n$ such that $P \cap S \neq \emptyset$ for all $P \in \mathcal{P}_{\mathcal{F}}$.
    Note that the paths $Q_i$ for $i \in [n]$ are arc disjoint and hence $|Q_i \cap S| = 1$ for each $i \in [n]$.
    In particular, this implies for every $i \in [n]$, the set $S$ contains at most one of the arcs $(v_i, v_{i,\text{false}})$, $(v_i, v_{i,\text{true}})$.
    Consider the assignment of the $x$-variables that sets $x_i = \text{true}$ for each $i \in [n]$ with $(v_i, v_{i,\text{false}}) \in S$ and that sets $x_i = \text{false}$ for all other $i$.
    By contradiction assume there is an assignment using the above truth values for $x$ and sets the $y$-variables such that $\phi(x, y)$ is false.
    Let $P \in \mathcal{P}_{\mathcal{F}}$ be the path corresponding to this assignment.
    Then $P \cap S \neq \emptyset$ implies there is $i \in [n]$ with $(v_i, v_{i,\text{false}}) \in S \cap P$ or $(v_i, v_{i,\text{true}}) \in S \cap P$.
    If $(v_i, v_{i,\text{false}}) \in S \cap P$, we have set $x_i = \text{true}$, which implies $(u_{i,\text{true}}, u'_{i,\text{true}}) \in P$, a contradiction.
    Symmetrically, if $(v_i, v_{i,\text{true}}) \in S \cap P$, we have set $x_i = \text{false}$, which implies $(u_{i,\text{false}}, u'_{i,\text{false}}) \in P$, a contradiction.
    We have constructed an assignment for the $x$-variables such that no assignment for the $y$-variables fulfills $\phi$.
\end{proof}

\begin{theorem}\label{thm:sigma-2-p-forbidden-pairs-identification}
    The following problem is $\Sigma\mathsf{^P_2}$-complete: Given a directed acyclic graph $D = (V, E)$ with two vertices $s, t \in V$, a set $\mathcal{F} \subseteq \binom{E}{2}$, a set $\bar{E} \subseteq E$, and $k \in \mathbb{N}$, find a set $S \subseteq E$ with $|S| \leq k$ such that $S \cup \bar{E}$ is identifying for $X := \{\mathbbm{1}_{P} : P \in \mathcal{P}_{\mathcal{F}}\}$.
    This holds even when restricted to instances where for each $e \in E$, there is $P \in \mathcal{P}_{\mathcal{F}}$ with $e \in P$ and, in addition, $F \cap \bar{E} \neq \emptyset$ for each $F \in \mathcal{F}$.
\end{theorem}

\begin{proof}
    We prove the theorem by reduction from \textsc{Forbidden Pairs $s$-$t$ Path Interdiction}.
    Consider an instance of this problem given by a digraph $D = (V, E)$, $s, t \in V$, $\mathcal{F} \subseteq \binom{E}{2}$, and $k \in \mathbb{N}$.
    By \cref{thm:sigma-2-p-forbidden-pairs-interdiction}, we can assume that $D$ is a directed acyclic graph and and for every $e \in E$ there is $P \in \mathcal{P}_{\mathcal{F}}$ with $e \in P$.
    For the instance of the problem described in \cref{thm:sigma-2-p-forbidden-pairs-identification}, we construct a digraph $D' = (V', E')$ as follows: 
    Let $D_0 = (V_0, E_0), D_1 = (V_1, E_1), D_2 = (V_2, E_2)$ be three copies of $D$ with disjoint arc sets $E_0, E_1, E_2$ and disjoint node sets $V_0, V_1, V_2$, except that $s_0 = s_1$ and $t_0 = t_1 = s_2$, where $s_0, s_1, s_2$ and $t_0, t_1, t_2$ denote the respective copies of $s$ and $t$.
    For $e \in E$ and $i \in \{0, 1, 2\}$, let further $e_i$ denote the corresponding copy of arc $e$ in $E_i$.
    We define the set $\mathcal{F}' \subseteq \binom{E'}{2}$ of forbidden arc pairs by including for each pair $\{e, e'\} \in \mathcal{F}$ the pair $\{e_2, e'_2\}$ in $\mathcal{F}'$ and furthermore for all $v \in V$ and all $e = (v, w), e' = (v, u') \in E$ with $e \neq e'$ we introduce the forbidden pairs $e_0, e'_2$ and $e_1, e'_2$.

    We define $s' := s_0 = s_1$ and $t' := t_2$ and let $\mathcal{P}'_{\mathcal{F}'}$ denote the set of $s'$-$t'$-paths $P$ in $D'$ such that~$|F \cap P| \leq 1$ for all $F \in \mathcal{F}'$. Let $X'$ be the set of incidence vectors of paths in $\mathcal{P}'_{\mathcal{F}'}$.
    Let further~$\bar{E} := E_2$. Note that this implies $F \cap \bar{E} \neq \emptyset$ for all $F \in \mathcal{F}$.
    
    We show that there is a set $S' \subseteq E'$ such that $|S'| \leq k$ and $S' \cup \bar{E}$ is identifying for $X'$ if and only if there is a set $S \subseteq E$ with $|S| \leq k$ such that $P \cap S \neq \emptyset$ for all $P \in \mathcal{P}_{\mathcal{F}}$. 
    For this, it is useful to define for any $P \in \mathcal{P}_{\mathcal{F}}$ and $i \in \{0, 1\}$, the path $P_i$ as the $s'$-$t'$-path in $D'$ consisting of the concatenation of the copy of $P$ in $D_i$ and the copy of $P$ in $D_2$.
    We first prove the following claim.

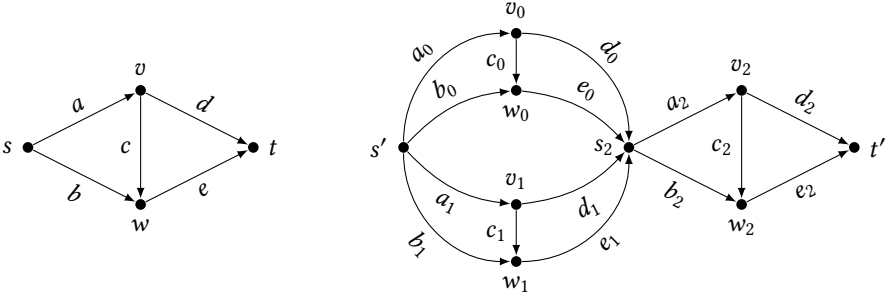
\begin{figure}
    \centering

    \begin{tikzpicture}
        \node[node, label=left:{$s$}] (s) at (0,0) {};
        \node[node, label=above:{$v$}] (v) at (1.5, 0.75) {};
        \node[node, label=below:{$w$}] (w) at (1.5, -0.75) {};
        \node[node, label=right:{$t$}] (t) at (3,0) {};
        
        \draw[-latex] (s) -- node[above, sloped] {$a$} (v);
        \draw[-latex] (s) -- node[below, sloped] {$b$} (w);
        \draw[-latex] (v) -- node[left] {$c$} (w);
        \draw[-latex] (v) -- node[above, sloped] {$d$} (t);
        \draw[-latex] (w) -- node[below, sloped] {$e$} (t);

        \node[node, label=left:{$s'$}] (s0) at (5,0) {};
        \path (s0)
            ++(1.5, 1.5) node[node, label=above:{$v_0$}] (v0) {}
            ++(0, -0.75)  node[node, label=below:{$w_0$}] (w0) {}
            ++(0, -1.5) node[node, label=above:{$v_1$}] (v1) {}
            ++(0, -0.75)  node[node, label=below:{$w_1$}] (w1) {}
            ++(1.5, 1.5) node[node, label=left:{$s_2$}] (s2) {};
        
        \draw[-latex, bend left=45] (s0) edge node[above, sloped] {$a_0$} (v0);
        \draw[-latex, bend left=20] (s0) edge node[above, sloped] {$b_0$} (w0);
        \draw[-latex] (v0) -- node[left] {$c_0$} (w0);
        \draw[-latex, bend left=45] (v0) edge node[above, sloped] {$d_0$} (s2);
        \draw[-latex, bend left=20] (w0) edge node[above, sloped] {$e_0$} (s2);
        
        \draw[-latex, bend right=20] (s0) edge node[below, sloped] {$a_1$} (v1);
        \draw[-latex, bend right=45] (s0) edge node[below, sloped] {$b_1$} (w1);
        \draw[-latex] (v1) -- node[left] {$c_1$} (w1);
        \draw[-latex, bend right=20] (v1) edge node[below, sloped] {$d_1$} (s2);
        \draw[-latex, bend right=45] (w1) edge node[below, sloped] {$e_1$} (s2);
        
        \path (s2)
            ++(1.5, 0.75) node[node, label=above:{$v_2$}] (v2) {}
            ++(0, -1.5)  node[node, label=below:{$w_2$}] (w2) {}
            ++(1.5, 0.75) node[node, label=right:{$t'$}] (t2) {};
        
        \draw[-latex] (s2) -- node[above, sloped] {$a_2$} (v2);
        \draw[-latex] (s2) -- node[below, sloped] {$b_2$} (w2);
        \draw[-latex] (v2) -- node[left] {$c_2$} (w2);
        \draw[-latex] (v2) -- node[above, sloped] {$d_2$} (t2);
        \draw[-latex] (w2) -- node[below, sloped] {$e_2$} (t2);
    \end{tikzpicture}
    
    \caption{Example for the construction in the proof of \cref{thm:sigma-2-p-forbidden-pairs-identification}. The left figure shows the digraph~$D$ and the right figure shows the corresponding digraph~$D'$ where $s' = s_0 = s_1$, $t_0 = t_1 = s_2$, and $t' = t_2$. The set of forbidden pairs $\mathcal{F} = \big\{ \{a, e\} \big\}$ in $D$ results in the forbidden pairs $\mathcal{F}' = \big\{\{a_2, e_2\}, \{a_0, b_2\}, \{a_1, b_2\}, \{b_0, a_2\}, \{b_1, a_2\}, \{c_0, d_2\}, \{c_1, d_2\}, \{d_0, c_2\}, \{d_1, c_2\} \big\}$ in $D'$.}
    \label{fig:sigma-2-p-forbidden-pairs-identification}
\end{figure}
    
    \begin{claim}
        For any $P \in \mathcal{P}_{\mathcal{F}}$, we have $P_0, P_1 \in \mathcal{P}'_{\mathcal{F}'}$.
        Conversely, for every $P' \in \mathcal{P}'_{\mathcal{F}'}$ there is $P \in \mathcal{P}_{\mathcal{F}}$ such that $P' = P_0$ or $P' = P_1$.
    \end{claim}
    \begin{proof}\renewcommand{\qedsymbol}{$\blacklozenge$}
        The fact that $P_0, P_1 \in \mathcal{P}'_{\mathcal{F}'}$ follows immediately from construction of $\mathcal{F}'$.
        For the second part of the claim, observe that any $s'$-$t'$-path in $D'$ consists of an $s_i$-$t_i$-path $Q'$ in $D_i$ for $i \in \{0, 1\}$, concatenated with an $s_2$-$t_2$-path $Q''$ in $D_2$.
        Moreover, by construction of $\mathcal{F}'$, the $s_2$-$t_2$-path $Q''$ must correspond to a path $P$ in $\mathcal{P}_\mathcal{F}$. 
        Now assume that $Q'$ uses an arc $e_i$ in $D_i$ that corresponds to an arc $e = (v, w) \in E \setminus P$, and choose $e$ so that $e_i$ is the first such are on $Q'$.
        Then $Q''$ contains arc~$e'_2$ for some $e' = (v, w') \in P$ and hence $P'$ contains both arcs of the forbidden pair $\{e_i, e'_2\} \in \mathbb{F}'$, a contradiction.
        Thus, both $Q'$ and $Q''$ correspond to copies of the same path $P$ and thus $P' = P_0$ or~$P' = P_1$.
    \end{proof}
    
    Now consider any set $S' \subseteq E'$ such that $S' \cup \bar{E}$ is identifying for $X'$ and $|S'| \leq k$.
    We let $S := \{e \in E : e_0 \in S' \text{ or } e_1 \in S'\}$.
    Note that $|S| \leq |S'| \leq k$.
    Now let $P \in \mathcal{P}_{\mathcal{F}}$.
    By the claim above, both $P_0$ and $P_1$ are in $\mathcal{P}_{\mathcal{F}'}$.
    Therefore $S' \cap (P_0 \Delta P_1) \neq \emptyset$, as $S' \cup \bar{E}$ is identifying and $P_0 \Delta P_1 \cap \bar{E} = \emptyset$.
    Thus, there is $e \in P$ such that $e_0 \in S'$ or $e_1 \in S'$ implying $e \in P \cap S$.

    For the converse, consider any $S \subseteq E$ with $|S| \leq k$ such that $P \cap S \neq \emptyset$ for all $P \in \mathcal{P}_{\mathcal{F}}$.
    Let~$S' := \{e_0 : e \in S\}$.
    Note that $|S'| = |S| = k$.
    Let $P', Q' \in \mathcal{P}'_{\mathcal{F}'}$ with $P' \neq Q'$.
    We show that $(P' \Delta Q') \cap (S' \cup \bar{E}) \neq \emptyset$.
    If $(P' \Delta Q') \cap \bar{E} \neq \emptyset$ we are done.
    Thus assume $(P' \Delta Q') \cap \bar{E} = \emptyset$, which is equivalent to $P' \cap E_2 = Q' \cap E_2$.
    Hence, by the claim above, there is $P \in \mathcal{P}_{\mathcal{F}}$ with $P' = P_0$ and $Q' = P_1$ (or the other way around).
    But then $P' \cap S' \neq \emptyset$ and $Q' \cap S' = \emptyset$ (or the other way around), showing that $(P' \Delta Q') \cap S' \neq \emptyset$. 
    We conclude that $S' \cup \bar{E}$ is identifying for $X'$.
\end{proof}

\begin{theorem}\label{thm:sigma-2-p-fixed-arcs}
    The following problem is $\Sigma^\mathsf{P_2}$-complete: Given a digraph $D = (V, E)$ with two vertices $s, t \in V$, a set $\bar{E} \subseteq E$, and $k \in \mathbb{N}$, find a set $S \subseteq E$ with $|S| \leq k$ such that $S \cup \bar{E}$ is identifying for $X := \{\mathbbm{1}_{P} : P \in \mathcal{P}\}$, where $\mathcal{P}$ is the set of $s$-$t$-paths in $D$.
    This holds even when restricted to instances where for each $e \in \bar{E}$, there is $P \in \mathcal{P}$ with $e \in P$.
\end{theorem}

\begin{proof}
    Consider an instance of the problem described in \cref{thm:sigma-2-p-forbidden-pairs-identification}, i.e., a digraph $D = (V, E)$, $s, t \in V$, $\mathcal{F} \subseteq \binom{E}{2}$, $\bar{E} \subseteq E$, and $k \in \mathbb{N}$.
    By \cref{thm:sigma-2-p-forbidden-pairs-identification}, we may assume that for each $e \in E$, there is $P \in \mathcal{P}_{\mathcal{F}}$ with $e \in P$.
    Without loss of generality, we can further assume that $F \cap F' = \emptyset$ for $F, F' \in \mathcal{F}$ (we can ensure this property by subdividing edges that appear in multiple pairs of $\mathcal{F}$). 
    We fix an arbitrary ordering of the arc pairs in $\mathcal{F}$, i.e., $\mathcal{F} = \big\{ \{e_1, e'_1\}, \{e_2, e'_2\}, \dots, \{e_{\ell}, e'_{\ell}\} \big\}$, where we can further assume that $e'_1, \dots, e'_{\ell} \in \bar{E}$ due to \cref{thm:sigma-2-p-forbidden-pairs-identification}.

    We construct an instance of the problem described in \cref{thm:sigma-2-p-fixed-arcs} with digraph $D' = (V', E')$, nodes $s', t' \in V'$, special edge set $\bar{E}' \subseteq E$.
    The node set $V'$ contains all nodes of $V$ and two additional nodes $s', t'$. The arc set $E'$ contains all arcs in $E$ that do not appear in  any forbidden pair from $\mathcal{F}$.
    In addition, 
    for each $i \in [\ell]$, we add the nodes $s_i, t_i, u_{i1}, u_{i2}, u_{i3}$ and $s'_i, t'_i, u'_{i1}, u'_{i2}, u'_{i3}$ and all arcs of the gadget $G_i$ depicted in \cref{fig:switch}, 
    where the nodes $v_i, v'_i$ and $w_i, w'_i$, respectively, are the heads and tails, respectively, of the arcs in the forbidden pair $\{e_i, e'_i\}$, i.e., $e_i = (v_i, w_i)$ and $e'_i = (v'_i, w'_i)$.
    We further add the arcs $(s', s_1)$, $(t_{\ell}, s)$, $(t, s'_{\ell})$, $(t'_1, t')$
    and the arcs
    $(t_i, s_{i + 1})$ and $(t'_{i+1}, s'_{i})$ for $i \in \{1, \dots, \ell-1\}$.
    See \cref{fig:sigma-2-p-fixed-arcs-construction} for a schematic overview of the construction.
    The set $\bar{E}'$ contains all arcs of $E'$ except the following:
    \begin{itemize}
        \item arcs $(v_i, u_{i1}), (u_{i1}, u_{i2}), (u_{i2}, w_i)$ if $e_i = (v_i, w_i)$ for some $i \in [\ell]$ with $e_i \notin \bar{E}$,
        \item arcs $(v'_i, u'_{i1}), (u'_{i1}, u'_{i2}), (u'_{i2}, w'_i)$ if $e'_i = (v'_i, w'_i)$ for some $i \in [\ell]$ with $e'_i \notin \bar{E}$,
        \item arcs $e \in E \setminus \bar{E}$ that do not appear in any pair $F \in \mathcal{F}$.
    \end{itemize}

We show that the constructed instance is a ``yes'' instance if and only if the instance of the problem from \cref{thm:sigma-2-p-forbidden-pairs-identification} is a ``yes'' instance.
We start with the following observation.

\begin{figure}
    \centering

    \begin{tikzpicture}
        \node[node, label=above:{$s'_i$}] (s0) at (0, 0) {};
        \node[node, label=below:{$t'_i$}] (t0)  at (0, -2.25) {};
        \node[node, label=below:{$s_i$}] (s1) at (0, -3.25) {};
        \node[node, label=above:{$t_i$}] (t1)  at (0, 1) {};

        \path (s0)
            ++(-3, +0.5) node[node, draw, fill=white, label=left:{$v_i$}] (v1) {}
            ++(+1, -0.5) node[node, label=above:{$u_{i1}$}] (u11)  {}
            ++(+1.25, -0.5) node[node, label=above:{$u_{i2}$}] (u12)  {}
            +(-2.25, -0.5) node[node, draw, fill=white, label=left:{$w_i$}] (w1) {}
            ++(0, -1.25) node[node, label=left:{$u_{i3}$}] (u13)  {};

        \path (s0)
            ++(+3, +0.5)  node[node, draw, fill=white, label=right:{$v'_i$}] (v2) {}
            ++(-1, -0.5) node[node, label=above:{$u'_{i1}$}] (u21)  {}
            ++(-1.25, -0.5) node[node, label=above:{$u'_{i2}$}] (u22)  {}
            +(+2.25, -0.5) node[node, draw, fill=white, label=right:{$w'_i$}] (w2) {}
            ++(0, -1.25) node[node, label=right:{$u'_{i3}$}] (u23)  {};

        \draw[-latex] (s0) edge[dashed] (u12);
        \draw[-latex] (s0) edge[dashed] (u22);
        \draw[-latex, bend left=20] (s1) edge[dashed] (u13);
        \draw[-latex, bend right=20] (s1) edge[dashed] (u23);
        
        \draw[-latex] (u13) edge[dashed] (t0);
        \draw[-latex] (u23) edge[dashed] (t0);
        \draw[-latex] (u11) edge[dashed] (t1);
        \draw[-latex] (u21) edge[dashed] (t1);

        \draw[-latex] (v1) edge (u11);
        \draw[-latex] (u11) edge (u12);
        \draw[-latex] (u12) edge[dashed] (u13);
        \draw[-latex] (u12) edge (w1);
        \draw[-latex, bend right=10] (u13) edge[dashed] (u21);

        \draw[-latex] (v2) edge (u21);
        \draw[-latex] (u21) edge (u22);
        \draw[-latex] (u22) edge[dashed] (u23);
        \draw[-latex] (u22) edge (w2);
        \draw[-latex, bend left=10] (u23) edge[dashed] (u11);
        
    \end{tikzpicture}
    
    \caption{Gadget $G_i$ used in the proof of \cref{thm:sigma-2-p-fixed-arcs}. This is a simplified version of the ``switch'' gadget used by \citet{fortune1980directed}. Dashed arcs represent arcs in $\bar{E}$.}
    \label{fig:switch}
\end{figure}
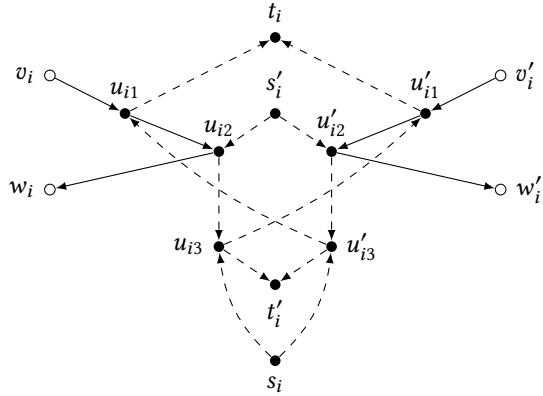

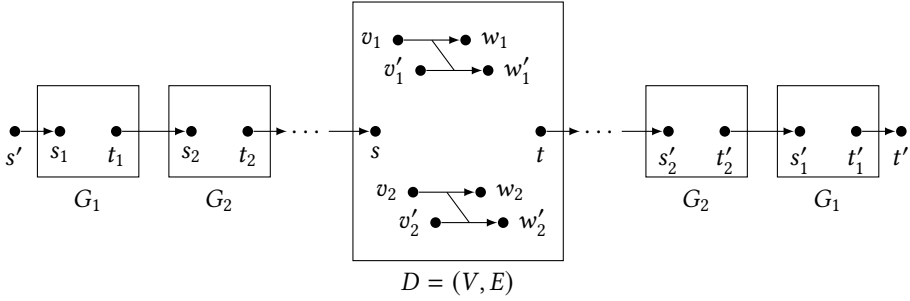
\begin{figure}
    \centering

    \begin{tikzpicture}
    
        \path node[node, label=below:{$s'$}] (sp) at (0, 0) {}
            ++(0.6, 0) node[node, label=below:{$s_1$}] (s1) {}
            +(0.375, -0.9) node {$G_1$}
            ++(0.75, 0) node[node, label=below:{$t_1$}] (t1) {}
            ++(1, 0) node[node, label=below:{$s_2$}] (s2) {}
            +(0.375, -0.9) node {$G_2$}
            ++(0.75, 0) node[node, label=below:{$t_2$}] (t2) {}
            ++(0.8, 0) node[inner sep=0.5mm] (d1) {$\dots$}
            ++(0.9, 0) node[node, label=below:{$s$}] (s) {}
            +(1.1, -2) node {$D = (V, E)$}
            ++(2.2, 0) node[node, label=below:{$t$}] (t) {}
            ++(0.8, 0) node[inner sep=0.5mm] (d2) {$\dots$}
            ++(0.9, 0) node[node, label=below:{$s'_{2}$}] (sp2) {}
            +(0.375, -0.9) node {$G_2$}
            ++(0.75, 0) node[node, label=below:{$t'_{2}$}] (tp2) {}
            ++(1, 0) node[node, label=below:{$s'_{1}$}] (sp1) {}
            +(0.375, -0.9) node {$G_1$}
            ++(0.75, 0) node[node, label=below:{$t'_{1}$}] (tp1) {}
            ++(0.6, 0) node[node, label=below:{$t'$}] (tp) {};
            
        \draw (s1) ++(-0.3, -0.6) rectangle +(1.35, 1.2);
        \draw (s2) ++(-0.3, -0.6) rectangle ++(1.35, 1.2);
        \draw (sp2) ++(-0.3, -0.6) rectangle ++(1.35, 1.2);
        \draw (sp1) ++(-0.3, -0.6) rectangle ++(1.35, 1.2);

        \draw (s) ++(-0.3, -1.7) rectangle ++(2.8, 3.4);
        
        \draw[-latex] (sp) edge (s1);
        \draw[-latex] (t1) edge (s2);
        \draw[-latex] (t2) edge (d1);
        \draw[-latex] (d1) edge (s);
        \draw[-latex] (t) edge (d2);
        \draw[-latex] (d2) edge (sp2);
        \draw[-latex] (tp2) edge (sp1);
        \draw[-latex] (tp1) edge (tp);

        \path (s) ++(0.3, 1.2) node[node, label=left:{$v_1$}] (v1) {}
                  +(0.9, 0) node[node, label=right:{$w_1$}] (w1) {}
                  ++(0.3, -0.4) node[node, label=left:{$v'_1$}] (vp1) {}
                  +(0.9, 0) node[node, label=right:{$w'_1$}] (wp1) {};

        \draw[-latex] (v1) edge (w1);
        \draw[-latex] (vp1) edge (wp1);
        \draw (v1) ++(0.45, 0) -- ++(0.3, -0.4);

        \path (s) ++(0.5, -0.8) node[node, label=left:{$v_2$}] (v2) {}
                  +(0.9, 0) node[node, label=right:{$w_2$}] (w2) {}
                  ++(0.3, -0.4) node[node, label=left:{$v'_2$}] (vp2) {}
                  +(0.9, 0) node[node, label=right:{$w'_2$}] (wp2) {};

        \draw[-latex] (v2) edge (w2);
        \draw[-latex] (vp2) edge (wp2);
        \draw (v2) ++(0.45, 0) -- ++(0.3, -0.4);

    \end{tikzpicture}
    
    \caption{Overview of the construction for proving \cref{thm:sigma-2-p-fixed-arcs}. The arcs $(v_i, w_i), (v'_i, w'_i)$ connected by a line represent a switch gadget $G_i$, where the nodes $v_i, v'_i, w_i, w'_i$ incident to arcs from the digraph $D$; the connection of the nodes $s_i, t_i, s'_i, t'_i$ from $G_i$ to nodes in other gadgets are depicted in the two boxed labeled $G_i$ on each side of the figure.}
    \label{fig:sigma-2-p-fixed-arcs-construction}
\end{figure}

\begin{lemma}\label{lem:disjoint-paths}
    Let $P$ be an $s'$-$t'$-paths in $D'$ and let $i \in [\ell]$ where $P \cap E(G_i)$ contains a path $\bar{P}$ starting at $s_i$ and a path $\bar{P}'$ ending at $t'_i$. Then exactly one of the following three cases holds true:
    \begin{enumerate}
        \item $\bar{P}$ corresponds to the node sequence $(s_i, u_{i3}, u'_{i1}, t_i)$ and $\bar{P}'$ corresponds to the node sequence $(s'_i, u'_{i2}, u'_{i3}, t'_i)$.\label{item:disjoint-paths-proper1}
        \item $\bar{P}$ corresponds to the node sequence $(s_i, u'_{i3}, u_{i1}, t_i)$ and $\bar{P}'$ corresponds to the node sequence $(s'_i, u_{i2}, u_{i3}, t'_i)$.\label{item:disjoint-paths-proper2}
        \item $\bar{P} = \bar{P}'$ is an $s_i$-$t'_i$-path corresponding to one of the four node sequences $(s_i, u_{i3}, t'_i)$, $(s_i, u'_{i3}, t'_i)$, $(s_i, u_{i3}, u'_{i1}, u'_{i2}, u'_{i3}, t'_i)$, or $(s_i, u'_{i3}, u_{i1}, u_{i2}, u_{i3}, t'_i)$.\label{item:disjoint-paths-not-proper}
    \end{enumerate}
\end{lemma}

\begin{proof}\renewcommand{\qedsymbol}{$\blacklozenge$}

    First assume that $\bar{P} \neq \bar{P}'$. 
    We show that we are in one of the first two cases.
    Because $\bar{P}$ does not leave $G_i$ at $t'_i$, it must contain one of the arcs $(u_{i3}, u'_{i1})$ or $(u'_{i3}, u_{i1})$.
    In the former case, $\bar{P}'$ can neither contain $u_{i3}$ nor $u'_{i1}$. 
    The only path entering $G_i$ at some node reachable from the outside and leaving it at $t'_i$ that contains neither of these nodes is the one with the node sequence $(s'_i, u'_{i2}, u'_{i3}, t'_i)$.
    Thus this path must be $\bar{P}'$, which also implies that $\bar{P}$ must leave $G_i$ at $t_i$ and correspond to the path with the sequence $(s_i, u_{i3}, u'_{i1}, t_i)$.
    Thus we are in case~\ref{item:disjoint-paths-proper1} of the lemma.
    By a symmetric argument, $(u'_{i3}, u_{i1}) \in \bar{P}$ implies that we are in case~\ref{item:disjoint-paths-proper2} of the lemma.
    
    Now assume that $\bar{P} = \bar{P}'$. Then $\bar{P}$ must be an $s_i$-$t'_i$-path.
    The paths listed in case~\ref{item:disjoint-paths-not-proper} of the lemma are exactly the four $s_i$-$t'_i$-paths in $G_i$.
\end{proof}

We say that an $s'$-$t'$-path $P'$ in $D'$ is \emph{proper} if, for every $i \in [\ell]$, the set $P' \cap G_i$ contains an~$s_i$-$t_i$-path and an $s'_i$-$t'_i$-path.
For a proper path $P$ and $i \in [\ell]$, we say that \emph{$P$ is in configuration~\ref{item:disjoint-paths-proper1} for $i$}, if case~\ref{item:disjoint-paths-proper1} of \cref{lem:disjoint-paths} holds, and otherwise we say that \emph{$P$ is in configuration~\ref{item:disjoint-paths-proper2} for $i$}.
Our next lemma shows that non-proper paths can be identified using the arcs from $\bar{E}$ alone and are thus not relevant for the solving the instance.

\begin{lemma}\label{lem:non-proper-paths}
    Let $P$ and $Q$ be two distinct $s'$-$t'$-paths in $D'$.
    If at least one of $P$ and $Q$ is not proper, then $(P \Delta Q) \cap \bar{E}' \neq \emptyset$.
\end{lemma}

\begin{proof}[Proof of \cref{lem:non-proper-paths}]\renewcommand{\qedsymbol}{$\blacklozenge$}
    Without loss of generality, assume that $P'$ is not proper and let $i \in [\ell]$ be minimal with the property that $P \cap G_i$ does not contain both an $s_i$-$t_i$-path and an $s'_i$-$t'_i$-path.
    By choice of $i$, each of the sets $P \cap G_j$ for $j < i$ contains an $s_j$-$t_j$-path $P_j$ and an $s'_j$-$t'_j$-path $P'_j$. Thus, the concatenation of $(s', s_1) \circ P_1 \circ (t_1, s_1) \circ \dots \circ  P_{i-1} \circ (t_{i-1}, s_i)$ must form the prefix $P[s', s_i]$ of $P$.
    Likewise, the concatenation of $(t'_i, s'_{i-1}) \circ P'_{i-1} \circ \dots \circ P'_1 \circ (t'_1, t')$ must form the suffix $P[t'_i, t']$ of $P$.
    Note that this implies that $P[s', s_i] \subseteq \bar{E}'$ and $P[t_i', t'] \subseteq \bar{E}'$.
    
    Because $P$ has a prefix ending at $s_i$ and a suffix starting at $t'_i$, the set $P \cap G_i$ must contain a path starting at $s_i$ $\bar{P}$ and a path $\bar{P}'$ starting at $t'_i$.
    By choice of $i$, we are in case \ref{item:disjoint-paths-not-proper} of \cref{lem:disjoint-paths}, i.e., $\bar{P} = \bar{P}'$ corresponds to one of the four $s_i$-$t'_i$-paths with the node sequences $(s_i, u_{i3}, t'_i)$, $(s_i, u'_{i3}, t'_i)$, $(s_i, u_{i3}, u'_{i1}, u'_{i2}, u'_{i3}, t'_i)$, or $(s_i, u'_{i3}, u_{i1}, u_{i2}, u_{i3}, t'_i)$.
    
    In particular, this implies that $P = P[s', s_i] \circ \bar{P} \circ P[t'_i, t']$. Moreover, we can conclude that either $\bar{P} \subseteq \bar{E}'$, or $\bar{P} \setminus \bar{E}'$ consists of exactly one arc, namely $(u_{i1}, u_{i2})$ or $(u'_{i1}, u'_{i2})$.
    Now assume by contradiction that $(P \Delta Q) \cap \bar{E}' = \emptyset$.
    Note that this implies $P \cap \bar{E}' \subseteq Q$.
    By the above, we are in one of the following three cases:
    \begin{itemize}
        \item $P \setminus \bar{E}' = \emptyset$: In this case, $P \subseteq Q$, which is only possible if $P = Q$, a contradiction.
        \item $P \setminus \bar{E}' = \{(u_{i1}, u_{i2})\}$: In this case, $P[s', u_{i1}] \subseteq Q$ and $P[u_{i2}, t'] \subseteq Q$. Furthermore, $(P \Delta Q) \cap \bar{E}' = \emptyset$ implies $(u_{i1}, t_i) \notin Q$ and therefore $(u_{i1}, u_{i2}) \in Q$, implying $P = Q$ again.
        \item $P \setminus \bar{E}' = \{(u'_{i1}, u'_{i2})\}$: This case is analogue to the preceding case. \qedhere
    \end{itemize}    
\end{proof}

The next two lemmas establish a direct correspondence between proper paths in $D'$ and elements of $\mathcal{P}_\mathcal{F}$, i.e., $s$-$t$-paths in $D$ respecting the forbidden pairs.

\begin{lemma}\label{lem:forbidden-pairs-to-proper-paths}
    Let $P \in \mathcal{P}_{\mathcal{F}}$.
    Let $P_1$ be the path resulting from $P$ by replacing all arcs $e_i$ and $e'_i$, respectively, for $i \in [\ell]$ by the subpaths
    $\{(v_i, u_{i1}), (u_{i1}, u_{i2}), (u_{i2}, w_i)\}$
    and 
    $\{(v'_i, u'_{i1}), (u'_{i1}, u'_{i2}), (u'_{i2}, w'_i)\}$, respectively.
    Then there is an $s'$-$s$-path $P_0$ and a $t$-$t'$-path $P_2$ in $D'$ such that $P' := P_0 \circ P_1 \circ P_2$ is a proper $s'$-$t'$-path in $D'$, with the additional property that for each $i \in [\ell]$, the proper path $P'$ is in configuration~\ref{item:disjoint-paths-proper2} for $i$ if and only if $e'_i \in P$.
\end{lemma}

\begin{proof}[Proof of \cref{lem:forbidden-pairs-to-proper-paths}]\renewcommand{\qedsymbol}{$\blacklozenge$}
    For each $i \in [\ell]$, let $Q_i = P_1 \cap G_i$.
    Because for each $|\{e_i, e'_i\} \cap P| \leq 1$ and by construction of $Q'$, we observe that $Q_i$ is either empty or corresponds to one of the two paths  $\{(v_i, u_{i1}), (u_{i1}, u_{i2}), (u_{i2}, w_i)\}$ or $\{(v'_i, u'_{i1}), (u'_{i1}, u'_{i2}), (u'_{i2}, w'_i)$\}.
   
    If $Q_i = \emptyset$ or $Q_i = \{(v_i, u_{i1}), (u_{i1}, u_{i2}), (u_{i2}, w_i)\}$, we let $\bar{P}_i$ and $\bar{P}'_i$, respectively, be the path corresponding to the node sequence $(s_i, u_{i3}, u'_{i1}, t_i)$ and $(s'_i, u'_{i2}, u'_{i3}, t'_i)$, respectively.
    Otherwise, we let $\bar{P}_i$ and $\bar{P}'_i$, respectively, be the path corresponding to the node sequence $(s_i, u'_{i3}, u_{i1}, t_i)$ and $(s'_i, u_{i2}, u_{i3}, t'_i)$, respectively.
    Node that in either case, $Q_i$, $\bar{P}_i$, and $\bar{P}'_i$ are node-disjoint and, moreover, entirely determined by $Q_i$.
    Letting $P_0$ be the concatenation $(s', s_1) \circ \bar{P}_1 \circ \dots \circ \bar{P}_{\ell} \circ (t_{\ell}, s)$ and $P_2$ be the concatenation $(t, s'_{\ell}) \circ \bar{P}'_{\ell} \circ \dots \circ \bar{P}_{1} \circ (t'_{1}, t')$ yields the desired paths such that $P' := P_0 \circ P_1 \circ P_2$ is proper.
\end{proof}

\begin{lemma}\label{lem:proper-paths}
    Let $P'$ be a proper $s'$-$t'$-path in $D'$ and define $$P := (Q' \cap E) \cup \{e_i : (v_i, u_{i1}) \in Q'\} \cup \{e'_i : (v'_i, u'_{i1}) \in Q'\}.$$
    Then $P$ is an $s$-$t$-path in $D$ and $P \in \mathcal{P}_\mathcal{F}$.
\end{lemma}

\begin{proof}[Proof of \cref{lem:proper-paths}]\renewcommand{\qedsymbol}{$\blacklozenge$}
    Because $P'$ is proper, $P' \cap G_i$ for each $i \in [\ell]$ must contain an $s_i$-$t_i$-path and an $s'_i$-$t'_i$-path corresponding to either case~\ref{item:disjoint-paths-proper1} or case~\ref{item:disjoint-paths-proper2} of \cref{lem:disjoint-paths}.
    In the former case, $P' \cap G_i$ can contain at most one other path, namely the one corresponding to the node sequence $(v_i, u_{i1}, u_{i2}, w_i)$.
    In the latter case, $P' \cap G_i$ can contain at most one other path, namely the one corresponding to the node sequence $(v'_i, u'_{i1}, u'_{i2}, w'_i)$.
    Hence, $P$ arises from $P'$ by replacing some $v_i$-$w_i$-subpaths with the corresponding arc $(v_i, w_i)$ and some $v'_i$-$w'_i$-subpaths with the corresponding arc $(v'_i, w'_i)$.
    Moreover, because $P'$ for each $i \in [\ell]$ contains at most one of the paths $(v_i, u_{i1}, u_{i2}, w_i)$ or $(v'_i, u'_{i1}, u'_{i2}, w'_i)$, the resulting path $P$ contains at most one of the arcs $e_i$ or $e'_i$.
    We conclude that $P \in \mathcal{P}_\mathcal{F}$.
\end{proof}

To complete the proof of \cref{thm:sigma-2-p-fixed-arcs}, consider any set $S \subseteq E$ such that $\bar{E} \cup S$ is an identifying instance for $X$ as specified in \cref{thm:sigma-2-p-forbidden-pairs-identification}.
Let
\begin{align*}
S' := (S \cap E) \cup \bigl\{(u_{i1}, u_{i2}) : e_i \in S \setminus \bar{E}\bigr\} \cup \bigl\{(u'_{i1}, u'_{i2}) : e'_i \in S \setminus \bar{E} \bigr\}
\end{align*}
and consider any two $s'$-$t'$-paths $P', Q'$ in $D'$ with $P' \neq Q'$.
We show that $(P' \Delta Q') \cap (\bar{E}' \cup S') \neq \emptyset$.
By \cref{lem:non-proper-paths}, we can assume that both $P'$ and $Q'$ are proper.
Then, by \cref{lem:proper-paths}, $P'$ and $Q'$ induce two $s$-$t$-paths $P, Q \in \mathcal{P}_{\mathcal{F}}$.
Because $\bar{E} \cup S$ is an identifying instance for $X$, there must be $e \in (P \Delta Q) \cap (\bar{E} \cup S)$.
Note that by construction of $P$ and $Q$ according to \cref{lem:proper-paths}, either $e \in P' \Delta Q'$, in which case we are done, or $e \in \{e_i, e'_i\}$ for some $i \in [\ell]$, in which case $P' \Delta Q'$ contains the edge $(u_{i1}, u_{i2})$ or $(u'_{i1}, u'_{i2})$, respectively.
Without loss of generality, assume $e = e_i$ and $(u_{i1}, u_{i2}) \in P' \Delta Q'$.
Then either $e \in \bar{E}$ and hence $(u_{i1}, u_{i2}) \in \bar{E}'$ or $e \in S$ and hence $(u_{i1}, u_{i2}) \in S'$.
We conclude that~$\bar{E}' \cup S'$ is identifying for $X' := \{\mathbbm{1}_{P}' : R \in \mathcal{P}'\}$, where $\mathcal{P}'$ is the set of $s'$-$t'$-paths in $D'$.

Conversely, consider any $S' \subseteq E'$ such that $\bar{E}' \cup S'$ is identifying for $X'$.
Let 
\begin{align*}
    S := (S' \cap E) & \cup \bigl\{e_i : e_i \notin \bar{E}, \{(v_i, u_{i1}), (u_{i1}, u_{i2}), (u_{i2}, w_i)\} \cap S' \neq \emptyset\bigr\}.
\end{align*}
Note that $|S| = |S'|$.
We show that $\bar{E} \cup S$ is identifying for $X$.
For this, let $P, Q \in \mathcal{P}_{\mathcal{F}}$.
If there is $i \in [\ell]$ with $e'_i \in P \Delta Q$ we are done because $e'_i \in \bar{E}$ by our initial assumption on the instance.
Thus we can assume that $e'_i \notin P \Delta Q$ for all $i \in [\ell]$.
Consider the two proper paths $P' = P_0 \circ P_1 \circ P_2 , Q' = Q_0 \circ Q_1 \circ Q_2$ guaranteed by \cref{lem:forbidden-pairs-to-proper-paths}.
Because for each  $i \in [\ell]$ we have that $e'_i \in P$ if and only if $e'_i \in Q$, \cref{lem:forbidden-pairs-to-proper-paths} implies that $P$ and $Q$ are in the same configuration for every $i$, i.e., $P_0 = Q_0$ and $P_2 = Q_2$. 
Therefore $P' \Delta Q' = P_1 \Delta Q_1$.

Because $S' \cup \bar{E}'$ is identifying for $X'$, there must be $e' \in S' \cup \bar{E}'$ with $e' \in P' \Delta Q' = P_1 \Delta Q_1$.
Without loss of generality, we assume that $e' \in P_1 \setminus Q_1$.
If $e' \in E$, then $e' \in P$ and $e' \in S$ by construction of $P_1$ and $S$.
If $e' \notin E$, then $e' \in \{(v_i, u_{i1}), (u_{i1}, u_{i2}), (u_{i2}, w_i)\}$ for some $i$. In that case, $e_i \in P$ by construction of $P_1$ and $e_i \in E_i \cup S$ by construction of $S$.
We thus conclude that $S$ is identifying for~$X$, completing the proof.
\end{proof}

We now show that \cref{thm:sigma-2-p-fixed-arcs} implies \cref{thm:paths-sigma-2-p}.

\begin{proof}[Proof of \cref{thm:paths-sigma-2-p}]
    Given an instance of the problem in \cref{thm:sigma-2-p-fixed-arcs}, we construct an instance of the problem in \cref{thm:paths-sigma-2-p} by replacing each $e = (u, v) \in \bar{E}$ by the following gadget: We remove $e$ and introduce two new vertices $u_e, v_e$ and the set of arcs $$B_e := \{(u, u_e), (u, v_e), (u_e, v_e), (u_e, v), (v_e, v)\}.$$ Let $D' = (V', E')$ be the digraph resulting from applying this modification to all $e \in \bar{E}$.

    Let $X' := \{\mathbbm{1}_{P} : P \in \mathcal{P}'\}$, where $\mathcal{P}'$ is the set of $s$-$t$-paths in $D'$.
    We show that there is an identifying set $S' \subseteq E'$ for $X'$ with $|S'| \leq k' := k + 2|\bar{E}|$ if and only if there is a set $S \subseteq E$ with $|S| \leq k$ such that $S \cup \bar{E}$ is identifying for $X$, proving the theorem.
    This follows immediately from the following two claims.
    \begin{claim}\label{clm:S-minus}
        Let $S \subseteq E$.
        If $S \cup \bar{E}$ is identifying for $X$, then $(S \setminus \bar{E}) \cup \{(u, u_e), (v_e, v) : e = (u, v) \in \bar{E}\}$ is identifying for $X'$.
    \end{claim}
    \begin{claim}\label{clm:Sprime}
        If $S' \subseteq E'$ is identifying for $X'$ then $|S' \setminus E| \geq 2 |\bar{E}|$ and $(S' \cap E) \cup \bar{E}$ is identifying for~$X$.
    \end{claim}

    Before we prove the claims, we make the following observations:
    \begin{itemize}
        \item If an $s$-$t$-path $P$ in $D'$ visits any of the two nodes $u_e, v_e$ for some $e = (u, v) \in \bar{E}$, then $P$ contains one of the three $u$-$v$-paths $\{(u, u_e), (u_e, v)\}$, $\{(u, u_e), (u_e, v_e), (v_e, v)\}$, and $\{(u, v_e), (v_e, v)\}$ of the subgraph induced by $B_e$ as a subpath.
        \item In particular, for any $P \in \mathcal{P}'$ and any  $e = (u, v) \in \bar{E}$, it holds that $P' \cap B_e$ is either $\emptyset$, $\{(u, u_e), (u_e, v)\}$, $\{(u, u_e), (u_e, v_e), (v_e, v)\}$, or $\{(u, v_e), (v_e, v)\}$ .
    \end{itemize}

    \begin{proof}[Proof of \cref{clm:S-minus}]\renewcommand{\qedsymbol}{$\blacklozenge$}
    To prove the first claim, let $S \subseteq E$ be an identifying set for $X$ and let $S' := (S \setminus \bar{E}) \cup \{(u, u_e), (v_e, v) : e = (u, v) \in \bar{E}\}$.
    Let $P', Q' \in \mathcal{P}'$ with $P' \neq Q'$.
    We distinguish two cases:
    \begin{itemize}
        \item Assume there is $e = (u, v) \in \bar{E}$ with $P' \cap B_e \neq Q' \cap B_e$.
        By the second observation above, $P' \cap B_e \neq Q' \cap B_e$ implies that $P' \Delta Q'$ contains $(u, u_e)$ or $(v_e, v)$ and hence $S' \cap (P' \Delta Q') \neq \emptyset$.
        \item Assume that $P' \cap B_e = Q' \cap B_e$ for all $e \in \bar{E}$.
        Let $P := P' \cap E \cup \{e \in \bar{E} : P' \cap B_e \neq \emptyset\}$ and 
        $Q := Q' \cap E \cup \{e \in \bar{E} : Q' \cap B_e \neq \emptyset\}$.
        Note that both $P$ and $Q$ are $s$-$t$-paths in $D$ and that $P \Delta Q = P' \Delta Q \subseteq E \setminus \bar{E}$.
        Because $S \cup \bar{E}$ is identifying for $X$, there must be $e \in (P \Delta Q) \cap (S \cup \bar{E}) = (P \Delta Q) \cap S = (P' \Delta Q') \cap S'$.\qedhere
    \end{itemize}
    \end{proof}

    \begin{proof}[Proof of \cref{clm:Sprime}]\renewcommand{\qedsymbol}{$\blacklozenge$}
    To prove the second claim, let $S' \subseteq E'$ be an identifying set for $X'$.
    
    We first show that $|S' \cap B_e| \geq 2$ for each $e \in \bar{E}$, which implies $|S' \setminus E| \geq 2|\bar{E}|$.
    For this, let $e = (u, v) \in \bar{E}$ and consider any $s$-$t$-path $P$ in $D$ containing $e$ (note that we can assume such a path exists by the final statement in \cref{thm:sigma-2-p-fixed-arcs}).
    Construct the $s$-$t$-path $P'$ in $D'$ from $P$ by replacing each $\bar{e} = (\bar{u}, \bar{v}) \in \bar{E}$ with the $\bar{u}$-$\bar{v}$-path $(\bar{u}, \bar{u}_{\bar{e}}), (\bar{u}_{\bar{e}}, \bar{v})$ in $D'$.
    Construct two further $s$-$t$-paths $Q', R'$ in $D'$ from $P'$ by replacing $(u, u_e), (u_e, v)$ with the $u$-$v$-path $(u, v_e), (v_e, v)$ or the $u$-$v$-path $(u, u_e), (u_e, v_e), (u_e, v)$, respectively. 
    Note that $P' \Delta Q', P' \Delta R', Q' \Delta R' \subseteq B_e$. 
    Thus, because $S'$ is identifying it must contain at least two arcs from $B_e$.
    
    Now let $S := S' \cap E$. We show that $S \cup \bar{E}$ is identifying for $X$.
    To see this, consider any two $s$-$t$-paths $P, Q \in \mathcal{P}$ in $D$ with $P \neq \emptyset$.
    If $(P \Delta Q) \cap \bar{E} \neq \emptyset$, we are done.
    Thus assume that $P \cap \bar{E} = Q \cap \bar{E}$.
    Construct $s$-$t$-paths $P', Q' \in \mathcal{P}'$, respectively, in $D'$ from $P$ and $Q$, respectively, by replacing each $\bar{e} = (\bar{u}, \bar{v}) \in \bar{E}$ in $P$ and $Q$, respectively, with the $\bar{u}$-$\bar{v}$-path $(\bar{u}, \bar{u}_{\bar{e}}), (\bar{u}_{\bar{e}}, \bar{v})$.
    Note that $P \Delta Q = P' \Delta Q' \subseteq E \setminus \bar{E}$.
    Hence $(P \Delta Q) \cap S = (P' \Delta Q') \cap S' \neq \emptyset$ because $S'$ is identifying for $X'$.
    \end{proof}

    This concludes the proof of \cref{thm:paths-sigma-2-p}.
\end{proof}

\section*{Acknowledgments}
This work was initiated at the Dagstuhl Seminar ``Dynamic Traffic Models in Transportation Science'' in May 2022.
We thank Richard Connors for raising the question that inspired this work, and all participants of the seminar for insightful discussions.
Jannik Matuschke was supported by the special research fund of KU Leuven (project CELSA/25/014).
Max Klimm's research was funded by the Deutsche Forschungsgemeinschaft (DFG, German Research
Foundation) under Germany's Excellence Strategy – The Berlin Mathematics
Research Center MATH+ (EXC-2046/1, EXC-2046/2, project ID: 390685689).

\bibliographystyle{ACM-Reference-Format}
\bibliography{control}

\end{document}